\documentclass[12pt]{amsart}

\usepackage{amsmath,amsthm,amsfonts,amssymb,latexsym,verbatim,bbm,hyperref}
\usepackage[shortlabels]{enumitem}
\usepackage{subcaption}
\usepackage{afterpage}

\DeclareRobustCommand{\SkipTocEntry}[5]{}

\allowdisplaybreaks

\setlist{itemsep=.5\baselineskip,topsep=.5\baselineskip}

\numberwithin{equation}{section}
\theoremstyle{plain}
\newtheorem{theorem}{Theorem}[section]
\newtheorem{thm}[theorem]{Theorem}
\newtheorem{lemma}[theorem]{Lemma}

\newtheorem{prop}[theorem]{Proposition}

\newtheorem{definition}[theorem]{Definition}
\newtheorem{defn}[theorem]{Definition}

\newtheorem{cor}[theorem]{Corollary}

\newtheorem{rmk}[theorem]{Remark}
\newtheorem{fact}[theorem]{Fact}

\newcommand{\arr}{\rightarrow}
\newcommand{\incl}{\hookrightarrow}

\newcommand{\C}{\mathbb{C}}
\newcommand{\Z}{\mathbb{Z}}

\newcommand{\eps}{\epsilon}

\newcommand{\Id}{\mathbbm{1}}

\newcommand{\mcB}{\mathcal{B}}

\newcommand{\mcF}{\mathcal{F}}
\newcommand{\mcG}{\mathcal{G}}
\newcommand{\mcH}{\mathcal{H}}
\newcommand{\mcI}{\mathcal{I}}

\newcommand{\mcO}{\mathcal{O}}

\newcommand{\mcS}{\mathcal{S}}

\newcommand{\mcU}{\mathcal{U}}

\newcommand{\mbN}{\mathbb{N}}

\usepackage{tikz}
\usetikzlibrary{positioning,calc}
\usetikzlibrary{decorations.pathmorphing,decorations.pathreplacing,patterns,shapes.symbols}
\definecolor{fillpaleyellow}{RGB}{255,250,205} 
\definecolor{fillbrownyellow}{RGB}{209,189,148} 
\definecolor{fillyellow}{RGB}{217,191,140}
\definecolor{fillgray}{RGB}{199,204,209}

\usepackage{braket}
\usepackage{color}
\usepackage[top=1.5in,bottom=1.5in,left=1.5in,right=1.5in]{geometry}

\newtheorem*{theorem*}{Theorem}

\newcommand{\norm}[1]{\left\lVert#1\right\rVert}

\newcommand{\tJ}{J}
\newcommand{\tM}{K}
\newcommand{\tR}{r}

\DeclareMathOperator{\Dehn}{Dehn}
\DeclareMathOperator{\Ar}{Area}
\newcommand{\tG}{\widetilde{G}_w}

\renewcommand\Re{\operatorname{Re}}
\renewcommand\Im{\operatorname{Im}}

\newcommand{\MIP}{\mathrm{MIP}}
\DeclareMathOperator{\PZKMIP}{\mathrm{PZK-MIP}}

\newcommand{\NTIME}{\mathrm{NTIME}}
\newcommand{\TIME}{\mathrm{TIME}}
\newcommand{\coNTIME}{\mathrm{coNTIME}}
\newcommand{\NEEXP}{\mathrm{NEEXP}}
\DeclareMathOperator{\poly}{poly}
\DeclareMathOperator{\NEXP}{NEXP}
\DeclareMathOperator{\NP}{NP}

\title[Lower bounds for computing non-local games to high precision]{Complexity lower bounds for computing the approximately-commuting operator value of non-local games to high precision} 

\author{Matthew Coudron}
\thanks{Institute for Quantum Computing, University of Waterloo, Waterloo,
    Canada. email: \texttt{mcoudron@uwaterloo.ca}}

\author{William Slofstra}
\thanks{Institute for Quantum Computing and Department of
    Pure Mathematics, University of Waterloo, Waterloo, Canada. email:
    \texttt{weslofst@uwaterloo.ca}}

\begin{document}

\begin{abstract}
We study the problem of approximating the commuting-operator value of a
two-player non-local game. It is well-known that it is $\NP$-complete to
decide whether the classical value of a non-local game is 1 or $1-
\epsilon$, promised that one of the two is the case.  Furthermore,
as long as $\epsilon$ is small enough, this result does not depend on 
the gap $\epsilon$.  In contrast, a recent result of Fitzsimons, Ji,
Vidick, and Yuen shows that the complexity of computing the quantum value grows
without bound as the gap $\epsilon$ decreases. In this paper, we show that this
also holds for the commuting-operator value of a game. Specifically, in the
language of multi-prover interactive proofs, we show that the power of
$\MIP^{co}(2,1,1,s)$ (proofs with two provers, one round, completeness
probability $1$, soundness probability $s$, and commuting-operator strategies)
can increase without bound as the gap $1-s$ gets arbitrarily small.

Our results also extend naturally in two ways, to perfect zero-knowledge
protocols, and to lower bounds on the complexity of computing the
approximately-commuting value of a game. Thus we get lower bounds on the
complexity class $\PZKMIP^{co}_{\delta}(2,1,1,s)$ of perfect zero-knowledge
multi-prover proofs with approximately-commuting operator strategies, as the
gap $1-s$ gets arbitrarily small. While we do not know any computable time
upper bound on the class $\MIP^{co}$, a result of the first author and Vidick
shows that for $s = 1-1/\poly(f(n))$ and $\delta = 1/\poly(f(n))$, the class
$\MIP^{co}_\delta(2,1,1,s)$, with constant communication from the provers, is
contained in $\TIME(\exp(\poly(f(n))))$.  We give a lower bound of
$\coNTIME(f(n))$ (ignoring constants inside the function) for this class, which
is tight up to polynomial factors assuming the exponential time hypothesis.
\end{abstract}

\maketitle

 \section{Introduction}
 
 Non-local games are a subject of converging interest for quantum information
 theory and computational complexity theory. A central question in both fields
 is the complexity of approximating the optimal winning probability of a
 non-local game.  Quantum mechanics allows non-local strategies in which the
 players share entanglement, and in quantum complexity theory we are interested
 in understanding the optimal winning probability over these entangled
 strategies.  Answering this question is necessary for understanding the power
 of multi-prover interactive proof systems with entangled provers and a
 classical verifier.
 
 For classical strategies (i.e. strategies without entanglement), it is NP-hard
 to decide whether a non-local game has winning probability $1$.  The PCP
 theorem implies that it is $\NP$-hard to decide whether a non-local game has
 winning probability $1$ or winning probability $1-\eps$, where $\eps$ is
 constant, promised that one of the two is the case \cite{AroLunMotSudSze98JACM,
 	AroSaf98JACM}. Therefore, for classical games, the complexity of computing the
 winning probability is the same for constant error as for zero error. 
 
 
 Two models for quantum strategies have historically been used when defining
 the entangled value of a nonlocal game: the tensor product model and the
 commuting-operator model. The optimal winning probability of a non-local game
 over tensor product strategies is called the quantum value, and optimal winning
 probability over all commuting-operator strategies is called the
 commuting-operator value. 
 
 A number of lower bounds on approximating the quantum value of a non-local
 game are known.  In particular, Ji has shown that it is $\NEXP$-hard to compute
 the quantum value of a non-local game with inverse polynomial precision, and
 $\NEEXP$-hard to compute the entangled value with inverse exponential precision
 \cite{Ji16}. Fitzsimons, Ji, Vidick, and Yuen continue this line of results
 by showing, roughly, that for any computable
 function $f(n): \mathbb{N} \to \mathbb{N}$, it is $\NTIME(\exp(f(n)))$ hard to compute the
 quantum value of a nonlocal game with $1/f(n)$ precision (here $n$ is the input size) \cite{FJVY18}. In particular, this implies that the quantum
 value of a game behaves very differently from the classical winning probability,
 since the complexity of computing the quantum value increases without bound
 as the required precision increases. 
 
 It is also natural to ask whether one might be able to approximate the
 commuting-operator value of a game efficiently.  The study of the
 commuting-operator value goes back to \cite{IKPSY08}, where it is shown that it
 is NP-hard to distinguish whether the commuting operator value is 1 or
 1-1/poly(n). The complexity of the commuting operator value does not seem to be
 explicitly studied in more recent work.
 
 In this paper, we look at lower bounds on the complexity of approximating the
 commuting-operator value of linear system nonlocal games, a type of nonlocal
 game closely connected with the theory of finitely-presented groups
 \cite{CLS16}. We show that group-theoretic methods can be used to lower bound
 the complexity of approximating the commuting-operator value of a linear system
 nonlocal game.  In particular we show that, just as with the quantum value of a
 game, the complexity of computing the commuting operator value of a non-local
 game to precision $\epsilon$ grows arbitrarily large as $\epsilon$ decreases.
 Because our results are based on group-theoretic methods, we observe that they
 naturally extend to lower bounds on approximately-commuting-operator strategies
 for games, a generalization of commuting-operator strategies in which Alice and
 Bob's strategies can interact slightly, but in such a way that the interaction
 is bounded by a parameter $\delta$. Thus we show:
 
 \begin{thm}\label{T:main_game}
 	There is a universal constant $k$ such that for every language $L \subset
 	A^*$ over a finite alphabet $A$ and contained in $\coNTIME(f(n))$, where
 	$f(n)$ is at least polynomial, there is a constant $C > 0$ and a
 	family of two-player non-local games $(\mcG_w)_{w \in A^*}$ of size
 	$\poly(n)$ and computable in $\poly(n)$-time, such that for any
 	$\delta = o(1/f(Cn)^{k})$, deciding
 	whether $\omega^{co}(\mcG_w) = 1$, or  
 	\begin{equation*}
 	\omega^{co}_{\delta}(\mcG_w) \leq 1 - 1/f(Cn)^k + O(\delta), 
 	\end{equation*}
 	promised that one of the two is the case, is as hard as deciding membership
 	of $w$ in $L$. 
 \end{thm}
Here $\omega^{co}(\mcG_w)$ is the commuting-operator value of $\mcG_w$, which
is the supremum over winning probabilities of all commuting-operator strategies
for $\mcG_w$. Simiarly, $\omega^{co}_{\delta}(\mcG_w)$ is the
$\delta$-commuting-operator value of $\mcG_w$, i.e. the supremum over winning
probabilities of all $\delta$-commuting-operator strategies (see Definition
\ref{defn::AC}). For any $\delta \geq 0$, we have 
\begin{equation*}
    \omega^{co}(\mcG_w) = \omega^{co}_{0}(\mcG_w) \leq \omega^{co}_{\delta}(\mcG_w)
        \leq 1,
\end{equation*}
so in particular, if $\omega^{co}(\mcG_w) = 1$, then $\omega^{co}_{\delta}(\mcG_w)=1$.
Thus Theorem \ref{T:main_game} gives a hardness result for approximating
$\omega^{co}_{\delta}(\mcG_w)$. If $\delta=0$, then we get a
hardness result for approximating the commuting-operator value of 
two-player non-local games. 
The proof of Theorem \ref{T:main_game} is given in Section \ref{S:mainproof}.
 
 The proof of Theorem \ref{T:main_game} relies on a deep group theory result of
 Sapir, Birget, and Rips, which shows that the acceptance problem for any Turing
 machine can be encoded in the word problem of a finitely-presented group, in
 such a way that the Dehn function of the group is equivalent to the running
 time of the Turing machine \cite{SBR02}. We then use \cite{Sl16} to embed this
 group into linear system non-local games. In the case that a word $w \in A^*$
 does not belong to $L$, the provers demonstrate this fact by showing that a
 certain word in the corresponding group is not equal to the identity. In this
 case, the representation of the group forms the proof that the word is not
 equal to the identity, and this representation is used to build the provers'
 quantum strategy. The reason that we use commuting-operator strategies in
 Theorem \ref{T:main_game}, and again in Theorem \ref{T:main_MIP} below, is that
 this representation might not be finite-dimensional. 
 
 Little is known about upper bounds on the complexity of computing the value of
 non-local games.  Most existing proposals for an algorithm are based on a
 hierarchy of semi-definite programs \cite{NPA07, NPA08NJP, DLTW08}.  It remains
 open whether such an algorithm can approximate the commuting-operator value of
 a game to any precision $\epsilon$ in finite time.  However, the first author
 and Vidick have shown that the SDP hierarchy of \cite{NPA07, NPA08NJP, DLTW08}
 can be used to estimate (with explicit convergence bounds) the optimal value of
 a non-local game over approximately-commuting strategies \cite{CV15}. 
 In particular \cite{CV15} gives an algorithm which, given a description of a
 non-local game as a truth-table of size $n$, can decide whether the game has
 commuting-operator value equal to 1, or has no $\delta$-commuting-operator
 strategy with winning probability higher than $1-\eps$, in time
$\poly\left(\log(1/\eps), n^{\poly(\ell, 1/\delta)}\right)$, where $\ell$ is the size of
the output set for the game. The games $\mcG_w$ in Theorem \ref{T:main_game}
have constant size output sets. Taking $\delta = o(1/f(Cn)^k)$ and $\eps =
1/f(Cn)^k - O(\delta) = O(1/f(Cn)^k)$, Theorem \ref{T:main_game} shows
that deciding whether $\omega^{co}(\mcG)=1$ or $\omega^{co}_{\delta}(\mcG)
\leq 1-\eps$ is $\coNTIME(f(n))$-hard. 
According to the exponential time hypothesis, we might expect that the best
deterministic upper bound for $\coNTIME(f(n))$ is $\TIME(2^{\poly(f(n))})$.
Thus, if we assume the exponential time hypothesis, the non-deterministic lower
bound in Theorem \ref{T:main_game} matches the deterministic upper bound in
\cite{CV15} up to polynomial factors (for families of games with a constant
number of outputs).
 
Results about the complexity of non-local games have direct and natural
implications for the power of multi-prover interactive proofs. Multi-prover
interactive proofs were originally defined and studied in a purely classical
setting.  A seminal result of Babai, Fortnow, and Lund, which studies the class
$\MIP$ of languages which admit a multi-prover interactive proof with
polynomial time verifier, states that $\MIP = \NEXP$. Once again, this equality
is independent of the completeness-soundness gap, as long as this gap is a
large enough constant. For entangled strategies, there are, a priori, two
analogs of the class $\MIP$ to consider, the class $\MIP^*$ of multi-prover
interactive proofs in which provers may use finite-dimensional entangled
strategies, and $\MIP^{co}$, the equivalent class with commuting-operator
strategies. A result of Ito and Vidick states that the class
$\MIP^*(4,1,1,1-1/\poly(n))$ with four provers, one round, completeness probability
1, and soundness probability 1-1/poly(n) contains $\NEXP$ \cite{IV12}.
Ji's result mentioned earlier for computing the quantum value of game shows
that with a sufficient number of provers $k$, $\MIP^*(k,1,1,1-1/\exp(n))$
contains NEEXP, in contrast again to the classical case \cite{Ji16}. Ji's
result is based on a compression theorem for non-local games, which also shows
that the problem of computing the quantum value of a game is complete for
$\MIP^*$.
 
Theorem \ref{T:main_game} can be translated into  lower bounds on
$\MIP^{co}_{\delta}$, the class of languages with a multiprover interactive
proof sound against approximately commuting strategies. Furthermore, these
lower bounds also apply to the class $\PZKMIP^{co}_{\delta}$ of languages which
admit a perfect zero knowledge multiprover interactive proof sound against
approximately commuting strategies. In a perfect zero knowledge interactive
proof the provers must reveal nothing to the verifier except the proven
statement itself. The formal definition of these two classes is given in
Definitions \ref{defn::mip} and \ref{defn::pzkmip}.
\begin{thm}\label{T:main_MIP}
 	There is a universal constant $k$ such that for any language $L$ in
 	$\NTIME(f(n))$, where $f(n)$ is at least polynomial, there is a constant
 	$C$ such that for any $\delta = o(1/f(Cn)^{k})$, 
 	\begin{equation*}
 	\overline{L} \in \PZKMIP^{co}_{\delta}(2,1,1,1-1/f(Cn)^k),
 	\end{equation*}
 	where $\overline{L}$ is the complement of $L$. 
\end{thm}
Note that, since the containment $\PZKMIP^{co}_{\delta} \subseteq
\MIP^{co}_{\delta}$ is immediate (see Definition \ref{defn::pzkmip}), Theorem
\ref{T:main_MIP} represents both a lower bound for $\PZKMIP^{co}_{\delta}$ and
for $\MIP^{co}_{\delta}$ itself.    Similarly to Theorem \ref{T:main_game},
when $\delta=0$, we get a lower bound on the class $\MIP^{co} := \MIP^{co}_0$
of multi-prover interactive proofs with commuting-operator strategies, which is
the direct analog of the complexity class $\MIP^*$ in the commuting operator
setting.
 
One reason we are interested in the class $\MIP^{co}_{\delta}$ is that the
algorithm of \cite{CV15} mentioned above gives a (deterministic) time upper
bound for $\MIP^{co}_{\delta}$, described in Theorem \ref{T:upperbound}.  In
contrast, no computable upper bounds for $\MIP^*$ or $\MIP^{co}$ are known.
Assuming $f(n)$ is at least exponential, combining Theorem \ref{T:main_MIP}
and the upper bound of \cite{CV15} gives the containments
\begin{align} \label{eq:maineq}
    \coNTIME(f(n)) &\subseteq \PZKMIP^{co}_{\delta}(2,1,1,1-1/f(Cn)^k) \\
    &\subseteq \MIP^{co}_{\delta}(2,1,1,1-1/f(Cn)^k) \nonumber \\
    & \subseteq \TIME(\exp(1 /\poly(\delta))) \nonumber,
\end{align}
where $C$ and $k$ are constants, and $\delta = o(1/f(Cn)^k)$.  Just as
for the decision problem in Theorem \ref{T:main_game}, if we assume the
exponential time hypothesis then we can consider the left hand side and right
hand side of Equation \ref{eq:maineq} above to be matching up to polynomial
factors. If we restrict $\PZKMIP^{co}$ and $\MIP^{co}$ to protocols with
constant-sized output sets, then this chain of inequalities holds as long as
$f(n)$ is at least polynomial.
 
Our results are complementary to the results of Fitzsimons, Ji, Vidick, and
Yuen, who show qualitatively similar lower bounds for computing the quantum
value of $k$-player games and for $\MIP^*(k,1,1,s)$, where $k \geq 15$.  Their
results show that $\MIP^*$ with $1/f(n)$ completeness-soundness gap contains
$\NTIME(2^{f(n)})$, matching the pattern seen in \cite{Ji16} for inverse
polynomial and inverse exponential gaps. In contrast, in our result the scaling
of the lower bound relative to the gap is weaker, requiring gap of order
$1/f(n)$ to get a lower bound of $\coNTIME(f(n))$, and applying to
commuting-operator strategies rather than quantum strategies. However, our
results apply to two-player protocols, while the results of \cite{FJVY18} apply
to protocols with 15 or more players.  That we get a lower bound of
$\coNTIME(f(n))$ rather than $\NTIME(2^{f(n)})$ can be explained by the fact
that our lower bound extends to $\MIP^{co}_{\delta}$, which, with the
restriction to protocols with constant-sized outputs, is contained in
$\TIME(2^{f(n)})$. Thus our results highlight the importance of considering
soundness to approximately-commuting strategies when seeking lower bounds on
$\MIP^*$ and $\MIP^{co}$. It seems to be an interesting open problem to
determine whether the improved bounds of \cite{FJVY18} can be done with
algebraic methods. 

\subsection{Acknowledgements}

We thank Zhengfeng Ji, Alex Bredariol Grilo, Anand Natarajan, Thomas Vidick, and Henry Yuen for helpful discussions.

MC was supported at the IQC by Canada's NSERC and the Canadian Institute for
Advanced Research (CIFAR), and through funding provided to IQC by the
Government of Canada and the Province of Ontario.  WS was supported by NSERC DG
2018-03968.
 
 \section{Group theory preliminaries}
 
 Recall that a \emph{finitely-presented group} is a group $G$ with a fixed
 presentation $G = \langle S : R \rangle$, meaning that $G$ is the quotient of
 the free group $\mcF(S)$ generated by a finite set $S$, by the normal subgroup
 generated by a finite set of relations $R \subseteq \mcF(S)$. If $G = \langle S
 : R \rangle$, and $R' \subseteq \mcF(S \cup S')$, then the notation $\langle G,
 S' : R' \rangle$ refers to the presentation $\langle S\cup S' : R \cup R'
 \rangle$.  A \emph{(group) word of length $k$} over the generators $S$ is a
 string $s_1^{a_1} \cdots s_k^{a_k}$ where $s_i \in S$ and $a_i \in \{\pm 1\}$
 for all $1 \leq i \leq k$. Such a word is said to be \emph{reduced} if $s_i =
 s_{i+1}$ implies that $a_i = a_{i+1}$ for all $1 \leq i \leq k-1$.  Every
 element $w \in \mcF(S)$ is represented by a unique reduced word, and the
 \emph{length $|w|$ of $w$} is defined to be the length of this reduced word. 
 The \emph{word problem} for $G$ is the problem of deciding whether the image of
 a given element $w \in \mcF(S)$ is equal to the identity in $G$, or in other
 words, whether the word is in the normal subgroup of $\mcF(S)$ generated by
 $R$. Since the reduced form of any non-reduced word over $S$ can be found in
 time linear in the length of that non-reduced word, we can ask that inputs to the word problem be represented either
 as reduced or non-reduced words without changing the problem. 

A \emph{(unitary) representation} of a group $G$ is a homomorphism $\phi : G
\arr \mcU(\mcH)$, where $\mcU(\mcH)$ is the unitary group of a Hilbert space
$\mcH$. If $G = \langle S : R \rangle$ is a finitely-presented group, then a
representation $\phi : G \arr \mcU(\mcH)$ can be specified by giving a
homomorphism $\widetilde{\phi} : \mcF(S) \arr \mcU(\mcH)$ such that
$\widetilde{\phi}(r) = 1$ for every $r \in R$. If $G$ is a group, then $\ell^2
G$ is the Hilbert space with Hilbert basis $\mcB = \{\ket{g} : g \in G\}$. This means
that every element of $\mcH$ is of the form $\sum_{g \in G} c_g \ket{g}$, where
$\sum_{g \in G} |c_g|^2 \leq +\infty$. Since every group $G$ acts on itself by
both left and right multiplication, $G$ also acts by left and right
multiplication on $\mcB$. Thus $G$ acts unitarily on $\ell^2 G$ by left
and right multiplication. The resulting representations $L,R : G \arr
\mcU(\ell^2 G)$ are called the \emph{left} and \emph{right regular
representations} of $G$, respectively. 

 If $w \in \mcF(S)$ is a word which is equal to the identity in $G$, we let $\Ar_G(w)$ be
 the minimum $t \geq 1$ such that 
 \begin{equation*}
 w = z_1 r_1^{a_1} z_1^{-1} \cdots z_t r_t^{a_t} z_t^{-1}
 \end{equation*}
 for some $r_1,\ldots,r_t \in R$, $z_1,\ldots,z_t \in \mcF(S)$, and
 $a_1,\ldots,a_t \in \{\pm 1\}$.\footnote{$\Ar_G(w)$ can also be defined as
 	the minimum number of regions in a van Kampen diagram with boundary word $w$,
 	and this is where the name comes from.} 
 The \emph{Dehn function $\Dehn_G$} of $G$ is the function $\mbN \arr \mbN$
 defined by
 \begin{equation*}
 \Dehn_G(n) = \max\{\Ar_G(w) : w \in \mcF(S) \text{ has } |w| \leq n \text{ and } w=1 \text{ in G}\}.
 \end{equation*}
 If the word problem of $G$ is decidable, then $\Dehn_G$ is computable.
 Conversely, the word problem of $G$ belongs to $\NTIME(\Dehn_G(n))$
 \cite{SBR02}. 
 An easy way to see that the complexity of the word problem is bounded by the
 Dehn function (albeit with the slightly worse upper bound of
 $\NTIME(\poly(\Dehn_G(n)))$) is through the following lemma:
 \begin{lemma}[\cite{Ge93}, Lemma 2.2]\label{lem::wordlengthbound}
 	Let $G = \langle S : R \rangle$ be a finitely-presented group, and let
 	$\ell$ be the length of the longest relation in $R$. If $w \in \mcF(S)$ is
 	equal to the identity in $G$ and $k = \Ar_G(w)$, then 
 	\begin{equation*}
 	w = z_1 r_1^{a_1} z_1^{-1} \cdots z_k r_k^{a_k} z_k^{-1}
 	\end{equation*}
 	where $r_1,\ldots,r_k \in R$, $z_1,\ldots,z_k \in \mcF(S)$, $a_1,\ldots,a_k
 	\in \{\pm 1\}$, and $|z_i| \leq k\ell + \ell + |w|$ for all $1 \leq i \leq k$. 
 \end{lemma}
 In general, the Dehn function can be much larger than the time-complexity of
 the word problem of $G$. However, Sapir, Birget, and Rips have shown that every
 recursive language can be reduced to the word problem of a finitely-presented
 group for which the Dehn function is polynomially equivalent to the
 time-complexity of the word problem. For the statement of the theorem, recall
 that two functions $T, T' : \mbN \arr \mbN$ are said to be
 \emph{(asymptotically) equivalent} if there are constants $C,C'$ such that
 $T(n) \leq C T'(Cn) + Cn + C$ and $T'(n) \leq C' T(C'n) + C'n + C'$ for all $n
 \geq 1$. 
 \begin{thm}[\cite{SBR02}, Theorem 1.3] \label{thm::SBR}
 	Let $A$ be a finite alphabet, and $L \subset A^*$ a language over $A$
 	contained in $\NTIME(T(n))$, where $T(n)$ is computable and $T(n)^4$ is at least superadditive (i.e.
 	$T(n+m)^4 \geq T(n)^4 + T(m)^4$. Then there exists a finitely-presented
 	group $G = \langle S : R \rangle$ and an injective function $\kappa : A^*
 	\arr \mcF(S)$, such that 
 	\begin{enumerate}[(a)]
 		\item $|\kappa(u)| = O(|u|)$ and $\kappa(u)$ is computable in time $O(|u|)$, 
 		\item $u \in L$ if and only if $\kappa(u) = 1$ in $G$, and 
 		\item $\Dehn_G(n)$ is bounded by a function equivalent to $T(n)^4$.  
 	\end{enumerate}
 \end{thm}
 
 A group over $\Z_2$ is a pair $(G,J)$ where $J$ is a central involution, i.e.
 an element of the center of $G$ with $J^2 = 1$. Usually we just write $G$ for
 the pair, and refer to $J = J_G$ in the same way we refer to the identity $1 =
 1_G$ of a group. When $J_G \neq 1_G$, it can be used as a substitute for $-1$.
 Theorem \ref{thm::SBR} implies that any recursive decision problem can be encoded in the
 word problem of a group. We want an embedding of this type where the word $w$
 is a central involution.  For this, we use the following trick:
 \begin{definition}\label{def::HNNtrick}
 	Let $G = \langle S : R\rangle$ be a finitely-presented group, and let
 	$x,\tJ,t$ be indeterminates not in $S$. Given $w \in \mcF(S)$, let
 	\begin{align*}
 	\tG := \langle G,x,\tJ,t\ :\quad   & \tJ^2 = 1, [g,\tJ] = 1 \text{ for all } g \in G, \\ 
 	& [x,\tJ]=1, [t,\tJ]=1, [t,[x,w]] = \tJ\rangle,
 	\end{align*}
 	where $[a,b] := a b a^{-1} b^{-1}$ is the group commutator. 
 \end{definition}
 Note that if $G$ is finitely-presented, then we only need to include the
 relations $[g,J]=1$ for $g$ in a generating set of $G$, and this gives
 a finite presentation of $\tG$.
 \begin{lemma}\label{lem::HNNtrick}
 	Given a group $G = \langle S : R \rangle$ and a word $w \in \mcF(S)$, let
 	$\tG$ be the group defined in Definition \ref{def::HNNtrick}. Then
 	\begin{enumerate}[(a)]
 		\item $\tJ$ is a central involution in $\tG$,
 		\item $w = 1$ in $G$ if and only if $\tJ=1$ in $\tG$, and
 		\item if $w=1$ in $G$ then $\Ar_{\tG}(\tJ) \leq 4 \Ar_G(w)+1$. 
 	\end{enumerate}
 \end{lemma}
 \begin{proof}
 	Part (a) is clear. For part (b), let 
 	\begin{equation*}
 	G' := \langle G,x,\tJ:  \tJ^2 = [x,\tJ] = [g,\tJ] = 1 \text{ for all } g \in G \rangle \footnote{Equivalently, in alternative language, this means that $G' := (G * \Z) \times \Z_2$, where $x$ is the generator of the $\Z$
 		factor, and $\tJ$ is the generator of the $\Z_2$ factor.},
 	\end{equation*}
 	The element $y =
 	[x,w]$ is equal to $1$ in $G'$ if and only if $w =1$. If $w \neq 1$ then
 	$y$ has infinite order. Hence the subgroup $\langle y, \tJ \rangle$ is
 	equal to $\Z \times \Z_2$ if $w \neq 1$, and $\Z_2$ if $w=1$. In both
 	cases, the homomorphism induced by $y \mapsto \tJ y$ and $\tJ \mapsto \tJ$
 	is an automorphism of this subgroup, and 
 	\begin{equation*}
 	\tG = \langle G', t : t y t^{-1} = \tJ y, t \tJ t^{-1} = \tJ\rangle
 	\end{equation*}
 	is the Higman-Neumann-Neumman (HNN) extension of $G'$ by this automorphism
 	(we refer to \cite[Chapter IV]{LS77} for the properties of HNN extensions).
 	As a result, $G'$ is a subgroup of $\tG$, and part (b) follows. 
 	
 	For part (c), if $w=1$ in $G$, then $\Ar_G(w^{-1}) = \Ar_G(w)$, so
 	\begin{equation*}
 	\Ar_{G'}([x,w]) \leq 2 \Ar_{G}(w) 
 	\end{equation*}
 	and similarly
 	\begin{equation*}
 	\Ar_{\tG}([t,[x,w]]) \leq 2 \Ar_{G'}([x,w]) \leq 4\Ar_{G}(w).
 	\end{equation*}
 	Thus we can use the relation $J = [t,[x,w]]$ to conclude that
 	$\Ar_{\tG}(J) \leq 4 \Ar_{G}(w) + 1$. 
 \end{proof}
 
 The last result we include in this section is a lemma which will be used to
 translate area calculations into bounds on distances between vectors in Hilbert
 spaces. If $u$ and $v$ are two vectors in a Hilbert space $\mcH$, we write $u
 \approx_{\eps} v$ to mean that $\norm{u - v} \leq \eps$. We use the standard
 terminology and notation of quantum information, so for instance, a state in a
 Hilbert space $\mcH$ is a unit vector $\ket{\psi}$ in $\mcH$. 
 \begin{defn}\label{D:bipartiterep}
 	Let $G = \langle S : R \rangle$ be a finitely-presented group. 
 	A \emph{$(\delta,\eps)$-bipartite representation of $G$ with respect to a
 		state $\ket{\psi}$ in a Hilbert space $\mcH$} is a pair of homomorphisms
 	$\Phi, \Phi' : \mcF(S) \arr \mcU(\mcH)$ such that
 	\begin{enumerate}[(i)]
 		\item $\Phi(r) \ket{\psi} \approx_{\eps} \ket{\psi}$ for all $r \in R$,
 		\item $\Phi(s)^{-1} \ket{\psi} \approx_{\eps} \Phi'(s) \ket{\psi}$ for all
 		$s \in S$, and
 		\item $\norm{[\Phi(s), \Phi'(t)] - \Id} \leq \delta$ for all
 		$s,t \in S$ (here $\Id$ represents the identity operator in $\mcU(\mcH)$). 
 	\end{enumerate}
 \end{defn}
 In Part (iii) and throughout this paper the notation $\norm{A}$ for an operator
 $A$ refers to the operator norm of $A$.  Part (i) of Definition
 \ref{D:bipartiterep} essentially says that $\Phi$ is an approximate
 representation of $G$ with respect to the state $\ket{\psi}$.  Parts (ii) and
 (iii) are less straightforward, but these conditions arise naturally in the
 theory of non-local games. 
 \begin{lemma}\label{lem::bipartiterep}
 	Let $(\Phi, \Phi')$ be a $(\delta,\eps)$-bipartite representation of a
 	finitely-presented group $G = \langle S : R \rangle$ with respect to a
 	state $\ket{\psi} \in \mcH$, and let $\ell$ be the length of the longest
 	relation in $R$. 
 	If $w \in \mcF(S)$ is equal to the identity in $G$, then 
 	\begin{equation*}
 	\Phi(w) \ket{\psi} \approx_{A(w)\cdot (\eps+\delta)} \ket{\psi},
 	\end{equation*}
 	where $A(w) \leq 5 \ell^2 \Ar_{G}(w)^2 + 2 \ell |w| \Ar_G(w)$.
 \end{lemma}
 \begin{proof}
 	If $r \in R$, then $\Phi(r) \ket{\psi} \approx_{\eps} \ket{\psi}$, and
 	consequently $\Phi(r)^{-1} \ket{\psi} \approx_{\eps} \ket{\psi}$. 
 	Thus for any $r \in R$, $z \in \mcF(S)$, and $a \in \{\pm 1\}$, 
 	\begin{align*}
 	\Phi(z r^a z^{-1}) \ket{\psi} & = \Phi(z) \Phi(r)^a \Phi(z)^{-1} \ket{\psi} \approx_{|z|\eps} \Phi(z) \Phi(r)^a \Phi'(z)^{-1} \ket{\psi} \\
 	& \approx_{|r||z| \delta} \Phi(z) \Phi'(z)^{-1} \Phi(r)^a \ket{\psi} \approx_{\eps} \Phi(z) \Phi'(z)^{-1} \ket{\psi} \\
 	& \approx_{|z| \eps} \Phi(z) \Phi(z)^{-1} \ket{\psi} = \ket{\psi}.
 	\end{align*}
 	We conclude that $\Phi(z r^a z^{-1}) \ket{\psi} \approx_{(2|z|+1)\eps + \ell |z| \delta} \ket{\psi}$. 
 	The result follows from Lemma \ref{lem::wordlengthbound}.
 \end{proof}
 
 \section{Approximately-commuting operator strategies and linear system games}
 
 A \emph{two-party Bell scenario} $(\mcI_A,\mcI_B,\mcO_A^*, \mcO_B^*)$ consists
 of finite input sets $\mcI_A, \mcI_B$, a finite set of outputs $\mcO_A^x$ for
 every $x \in \mcI_A$, and a finite set of outputs $\mcO_B^y$ for every $y \in
 \mcI_B$.\footnote{The sets $\mcO_A^x$ and $\mcO_B^y$ are often assumed to be
 	independent of the inputs $x$ and $y$. However, this assumption is not
 	essential, since we can make the output sets independent of the input sets by
 	adding filler answers to make all output sets the same size, and stipulating
 	that Alice and Bob lose if they output one of the filler answers.  When
 	working with linear system games, it is more convenient to have the output sets
 	depend on the inputs.}
 The number of outputs in a Bell scenario is the maximum of
 $|\mcO_A^x|$ and $|\mcO_B^y|$ over $x \in \mcI_A$ and $y \in \mcI_B$. A
 \emph{two-player non-local game} consists of a Bell scenario
 $(\mcI_A,\mcI_B,\mcO_A^*, \mcO_B^*)$, a function $V(\cdot,\cdot |x,y) :
 \mcO_A^x \times \mcO^y_B \arr \{0,1\}$ for every $x \in \mcI_A$ and $y \in
 \mcI_B$, and a probability distribution $\pi$ on $\mcI_A \times \mcI_B$. In the
 operational interpretation of the game, the referee sends players Alice and Bob
 inputs $x \in \mcI_A$ and $y \in \mcI_B$ with probability $\pi(x,y)$, the
 players reply with outputs $a \in \mcO_A^x$ and $b \in \mcO_B^y$, and the
 players win if and only if $V(a,b|x,y) = 1$. 
 
 In a non-local game, the players are not usually allowed to communicate while
 the game is in progress. Thus, in a quantum strategy for a game, it's assumed
 that each player determines their output by measuring their own local system.
 Locality can be enforced in one of two ways: by requiring that the joint system
 is the tensor product of the subsystems, or by requiring that measurement
 operators for different players commute with each other. Strategies of the
 former type are called tensor-product strategies, while strategies of the
 latter type are called commuting-operator strategies. Tensor-product strategies
 are commuting-operator strategies by definition, and finite-dimensional
 commuting-operator strategies can be turned into equivalent tensor-product
 strategies. In infinite dimensional Hilbert spaces, there are
 commuting-operator strategies for which the corresponding correlations do not
 have a tensor-product model \cite{Sl16}. However, it's still an open question
 as to whether all correlations arising from commuting-operator strategies can
 be realized as a limit of tensor-product strategies. By a theorem of Ozawa,
 this question is equivalent to the Connes embedding problem.  In \cite{Oz13b,
 	CV15}, the notion of a quantum strategy has been generalized to
 approximately-commuting strategies, where Alice and Bob's systems are allowed
 to interact slightly. 
 In this paper, we focus on the case of approximately-commuting operator
 strategies.  Unlike \cite{CV15}, we use projective measurements rather than the
 more general POVM measurements in this definition. We refer to Remark
 \ref{rmk::povm} for some of the consequences of this difference. 
 \begin{defn}\label{defn::AC}
 	A \emph{$\delta$-approximately-commuting operator strategy $\mcS$} (or
 	\emph{$\delta$-AC operator strategy} for short) for a Bell scenario 
 	$(\mcI_A,\mcI_B,\mcO_A^*,\mcO_B^*)$ consists of a Hilbert space $\mcH$, a
 	projective measurement $\{P^x_a\}_{a \in \mcO_A^x}$ on $\mcH$ for every $x
 	\in \mcI_A$, a projective measurement $\{Q^y_b\}_{b \in \mcO_B^y}$ on
 	$\mcH$ for every $y \in \mcI_B$, and a state $\ket{\psi} \in \mcH$ such
 	that
 	\begin{equation*}
 	\norm{P^x_a Q^y_b - Q^y_b P^x_a} \leq \delta
 	\end{equation*}
 	for all $(x,y) \in \mcI_A \times \mcI_B$ and $(a,b) \in \mcO_A^x \times
 	\mcO_B^y$. 
 	A \emph{$\delta$-approximately-commuting quantum} (or
 	\emph{$\delta$-AC quantum}) \emph{strategy} is a $\delta$-AC operator strategy in
 	which $\mcH$ is finite-dimensional. 
 	
 	Let $\mcG = (\mcI_A,\mcI_B,\mcO_A^*,\mcO_B^*,V,\pi)$ be a non-local game.
 	The \emph{winning probability of $\mcG$ with strategy $\mcS$} is
 	\begin{align*}
 	\omega(\mcG;\mcS) = \left|\sum_{x \in \mcI_A, y \in \mcI_B} \pi(x,y) \sum_{a \in \mcO_A, b \in \mcO_B} V(a,b|x,y) \bra{\psi} P^x_a Q^y_b \ket{\psi} \right|.
 	\end{align*}
 	The $\delta$-AC operator value $\omega^{co}_\delta(\mcG)$ (resp.
 	$\delta$-AC quantum value $\omega^*_{\delta}(\mcG)$) of $\mcG$ is defined
 	to be the supremum of $\omega(\mcG;\mcS)$ across $\delta$-AC operator
 	strategies (resp. $\delta$-AC quantum strategies). 
 \end{defn}
 With this definition, a \emph{commuting-operator strategy} is simply a $0$-AC
 operator strategy, and the usual \emph{commuting-operator value} of a game is
 $\omega^{co}(\mcG) := \omega^{co}_0(\mcG)$. 
 Since commuting-operator strategies are the same as tensor product strategies
 in finite dimensions, a \emph{(tensor-product) quantum strategy} is simply a
 $0$-AC quantum strategy, and the usual \emph{quantum value} of a game is
 $\omega^*(\mcG) := \omega^*_0(\mcG)$. 
 Note that when $\delta=0$, the absolute value can be dropped
 in the definition of $\omega(\mcG;\mcS)$. When $\delta > 0$, the values
 $\bra{\psi} P^x_a Q^y_b \ket{\psi}$ can be complex, and the absolute value is
 necessary. 
 This also means that $\omega(\mcG,\mcS)$ cannot necessarily be interpreted
 as a probability when $\mcS$ is approximately but not exactly commuting.
 
 We look at a specific class of non-local games called linear system games. Let
 $Mx =c$ be an $m \times n$ linear system over $\Z_2$, so $M \in \Z_2^{m \times
 	n}$ and $c \in \Z_2^m$. For each $1 \leq i \leq m$, let $V_i = \{ 1 \leq j \leq
 n : M_{ij} \neq 0\}$. The linear system game $\mcG_{Mx=c}$ is the non-local
 game with 
 \begin{equation*}
 \mcI_A = \{1,\ldots,m\},\quad \mcI_B = \{1,\ldots,n\}, 
 \end{equation*}
 \begin{equation*}
 \mcO_A^i = \left\{a \in \Z_2^{V_i} : \sum_{j \in V_i} a_j = c_i \right\}, 
 \quad \mcO_B^j = \Z_2,
 \end{equation*}
 \begin{equation*}
 V(a,b|i,j) = \begin{cases} 1 & j \not\in V_i \text{ or } a_j = b \\
 0 & \text{otherwise} 
 \end{cases}, 
 \end{equation*}
 and $\pi$ the uniform distribution over pairs $(i,j)$ such that $j \in V_i$.
 In other words, Alice receives the index $i$ of an equation and Bob receives
 the index $j$ of a variable, chosen uniformly at random from pairs $(i,j)$ with
 $j \in V_i$. Alice replies with a satisfying assignment to the variables which
 appear in the $i$th equation, and Bob replies with an assignment for the $j$th
 variable.  The players win if Alice and Bob both give the same assignment to
 variable $j$. 
 
 For linear system games, it is often convenient to express strategies in terms
 of observables, rather than measurement operators (see, for instance,
 \cite{CM14, CLS16}). If $\mcS = (\mcH, \{P^{i}_a\}_{a \in \mcO_A^i}, \{Q^j_b\}_{b \in
 	\Z_2}, \ket{\psi})$ is a $\delta$-AC strategy for $\mcG_{Mx=c}$, the corresponding
 observables are 
 \begin{equation}\label{E:observablesA}
 A_{ij} := \sum_{a \in \mcO_A^i} (-1)^{a_j} P^i_a \text{ for } 1 \leq i \leq m,
 j \in V_i, 
 \end{equation}
 and
 \begin{equation}\label{E:observablesB}
 B_j := Q^j_0 - Q^j_1 \text{ for } 1 \leq j \leq n.
 \end{equation}
 These operators are $\pm 1$-valued observables (meaning, self-adjoint unitary
 operators) satisfying the equations
 \begin{equation}\label{E:observables1}
 \prod_{j} A_{ij}^{M_{ij}} = (-\Id)^{c_i} \text{ for all } 1 \leq i \leq m, 
 \end{equation}
 \begin{equation}\label{E:observables2}
 [A_{ij},A_{ij'}]=\Id \text{ whenever } j,j' \in V_i \text{ for some } 1 \leq i \leq m, \text{ and}
 \end{equation} 
 \begin{equation}\label{E:observables3}
 \norm{[A_{ij},B_k] - \Id} \leq 2^{|V_i|+1} \delta 
 \text{ for all }1 \leq i \leq m, j \in V_i,\text{ and }1 \leq k \leq n. 
 \end{equation}
We can recover the projections $P^i_a$, $a \in \mcO_A^i$, and $Q^j_b$,
$b \in \mcO_B^j$, from the observables $A_{ik}$ and $B_j$ via the formulas
\begin{equation}\label{E:observables4}
 	P_a^i = \prod_{k \in V_i} \left(\frac{\Id + (-1)^{a_k} A_{ik}}{2}\right)
    \text{ and } Q^j_b = \frac{1 + (-1)^b B_j}{2}.
\end{equation}

 We define \emph{bias of strategy $\mcS$} to be
 \begin{equation*}
 \beta(\mcG_{Mx=c}; \mcS) := \sum_{1 \leq i \leq m} \sum_{j \in V_i}
 \pi(i,j) \bra{\psi} A_{ij} B_j \ket{\psi}.
 \end{equation*}
 It is not hard to see that 
 \begin{equation*}
 \omega(\mcG_{Mx=c}; \mcS) = \frac{1}{2} | \beta(\mcG_{Mx=c},\mcS) + 1|,
 \end{equation*}
 so we can work with the winning probability using observables as well.
 
 
 It follows from \cite{CLS16}
 that when $\delta=0$, perfect commuting-operator strategies of
 $\mcG_{Mx=c}$ can be understood using the following group.
 \begin{defn}\label{def::solutiongroup}
 	Let $Mx=c$ be an $m \times n$ linear system over $\Z_2$. Then
 	the \emph{solution group} of the system is the finitely presented
 	group $\Gamma_{Mx=c}$ generated by $x_1,\ldots,x_n, J$, and satisfying relations
 	\begin{enumerate}
 		\item $[x_i,J] = x_i^2 = J^2 = 1$ for all $1 \leq i \leq n$,
 		\item $\prod_{j} x_j^{M_{ij}} = J^{c_i}$ for all $1 \leq j \leq m$, and
 		\item $[x_j, x_k ] = 1$ if there is some $1 \leq i \leq m$ with $M_{ij}, M_{ik} \neq 0$.
 	\end{enumerate}
 	We consider  $\Gamma = \Gamma_{Mx=c}$ to be a group over $\Z_2$ with
 	$J_{\Gamma}$ equal to the generator $J$. 
 \end{defn}
 
 In particular, we can characterize when the optimal winning probability of the
 game is equal to $1$ using this group.
 \begin{theorem}[\cite{CLS16,SV17}]\label{thm::perfectstrats}
 	Let $Mx=c$ be a linear system over $\Z_2$. Then
 	\begin{enumerate}[(a)]
 		\item $\omega^{co}(\mcG_{Mx=c}) = 1$ if and only if $J \neq 1$ in $\Gamma_{Mx=c}$, and
 		\item $\omega^*(\mcG_{Mx=c})=1$ if and only if $J$ is non-trivial in approximate
 		representations of $\Gamma_{Mx=c}$. 
 	\end{enumerate}
 \end{theorem}
 For the definition of \emph{non-trivial in approximate representations}, we
 refer to \cite{SV17}. 
 
 Near-perfect finite-dimensional strategies of $\mcG_{Mx=c}$
 correspond to approximate representations of $\Gamma_{Mx=c}$ \cite{SV17}. 
 We want to develop this theory when $\delta > 0$. 
 \begin{prop}\label{prop::approxstrategy}
 	Let $Mx=c$ be an $m \times n$ linear system, let $V_i := \{1 \leq j \leq n:
 	M_{ij} \neq 0\}$, let $\tR := \max_i |V_i|$ be the maximum number
 	of non-zero entries in any row, and let $\tM := \sum_{i=1}^m |V_i|$ be the number
 	of non-zero entries in $M$. Suppose $\mcS = (\mcH, \{P^x_a\}, \{Q^y_b\},
 	\ket{\psi})$ is a $\delta$-AC operator strategy with $\omega(\mcG_{Mx=c};
 	\mcS) \geq 1-\eps$ for some $\eps,\delta \geq 0$. Let $A_{ij}$, $B_k$ be
 	the corresponding observables defined in Equations \eqref{E:observablesA}
 	and \eqref{E:observablesB}. Then
 	\begin{enumerate}[(a)]
 		\item $A_{ij} \ket{\psi} \approx_{2 \sqrt{\tM (\eps + 2^{\tR-1} \delta)}} B_j \ket{\psi}$
 		for all $1 \leq i \leq m$ and $j \in V_i$, 
 		\item $\prod_{j=1}^m B_{ij}^{M_{ij}} \ket{\psi} \approx_{2 \tR \sqrt{\tM (\eps + 2^{\tR-1} \delta)} + \binom{\tR}{2} 2^{\tR+1} \delta} (-\Id)^{c_i} \ket{\psi}$
 		for all $1 \leq i \leq m$, and
 		\item $[B_j, B_k] \ket{\psi} \approx_{8 \sqrt{\tM (\eps + 2^{\tR-1} \delta)} + 6\cdot  2^{\tR+1} \delta} \ket{\psi}$ whenever there is $1
 		\leq i \leq m$ with $j,k \in V_i$.
 	\end{enumerate}
 \end{prop}
 \begin{proof}
 	For part (a), any two unit vectors $\ket{\psi}$ and $\ket{\phi}$ satisfy
 	$\ket{\psi} \approx_2 \ket{\phi}$, so we can assume that $\eps + 2^{\tR-1}
 	\delta \leq 1$.  Write $\beta$ for $\beta(\mcG_{Mx=c}, \mcS)$, and observe
 	that
 	\begin{align*}
 	|2 \Im \beta| & = |\beta - \overline{\beta}| = \left| \sum_{i,j} \pi(i,j)
 	\bra{\psi} A_{ij} B_j - B_j A_{ij} \ket{\psi} \right| \\
 	& \leq \sum_{i,j} \pi(i,j) \norm{A_{ij} B_j - B_j A_{ij}} \leq 2^{\tR+1} \delta
 	\end{align*}
 	by Equation \eqref{E:observables3}. Since $\omega(\mcG_{Mx=c};\mcS) \geq 1-\eps$,
 	we have that
 	\begin{equation*}
 	(1 -\eps)^2 \leq \left| \frac{\beta + 1}{2} \right|^2 = 
 	\frac{(\Re \beta + 1)^2 + (\Im \beta)^2}{4} \leq \frac{(\Re \beta+1)^2 + 
 		\left(2^{\tR} \delta\right)^2}{4}.
 	\end{equation*}
 	Since $A_{ij} \ket{\psi}$ and $B_j \ket{\psi}$ are unit vectors, $-1 \leq
 	\Re \beta \leq 1$, and in particular $\Re\beta + 1 \geq 0$. Thus 
 	\begin{equation*}
 	\Re \beta + 1 \geq \sqrt{ 4(1 - \eps)^2 - \left(2^{\tR} \delta \right)^2}
 	= \sqrt{(2 - 2\eps - 2^{\tR} \delta)(2 - 2\eps + 2^{\tR} \delta)}
 	\geq 2 - 2 \eps - 2^{\tR} \delta, 
 	\end{equation*}
 	where the last inequality holds because of the assumption
 	$2 \eps + 2^{\tR} \delta \leq 2$. We conclude that $\Re \beta \geq 1 - 2
 	\eps - 2^{\tR} \delta$, or $1 - \Re \beta \leq 2 \eps + 2^{\tR} \delta$. 
 	
 	Now $\pi(i,j) = 1/\tM$ for all $1 \leq i \leq m$, $j \in V_i$, so 
 	\begin{equation*}
 	1 - \Re \beta = 
 	\frac{1}{\tM} \sum_{i,j} (1 - \Re \bra{\psi} A_{ij} B_j \ket{\psi}) 
 	\leq 2 \eps + 2^{\tR} \delta.
 	\end{equation*}
 	Since $\Re \bra{\psi} A_{ij} B_j \ket{\psi} \leq 1$, we have that 
 	$1 - \Re \bra{\psi} A_{ij} B_j \ket{\psi} \leq 2 \tM (\eps + 2^{\tR - 1} \delta)$
 	for all $1 \leq i \leq m$ and $j \in V_i$.
 	So
 	\begin{equation*}
 	\norm{A_{ij}\ket{\psi} -  B_j \ket{\psi}}^2 = 2 - 2 \Re \bra{\psi} A_{ij} B_j \ket{\psi} \leq 4 \tM (\eps + 2^{\tR-1} \delta), 
 	\end{equation*}
 	finishing the proof of part (a). 
 	
 	For parts (b) and (c), let $\tau = 2 \sqrt{\tM (\eps + 2^{\tR-1} \delta)}$. Given
 	$1 \leq i \leq m$, let $V_i = \{j_1,\ldots,j_k\}$, where $1 \leq j_1 <
 	\ldots < j_k \leq n$. Then
 	\begin{align*}
 	B_{j_1} \cdots B_{j_k} \ket{\psi} \approx_{\tau} B_{j_1} \cdots B_{j_{k-1}}
 	A_{i j_k} \ket{\psi} \approx_{(k-1) 2^{\tR+1} \delta} A_{i j_k}
 	B_{j_1} \cdots B_{j_{k-1}} \ket{\psi}.
 	\end{align*}
 	Continuing this pattern, we see that
 	\begin{align*}
 	B_{j_1} \cdots B_{j_k} \approx_{k \tau + \binom{k}{2} 2^{\tR+1} \delta} 
 	A_{ij_k} A_{ij_{k-1}} \cdots A_{ij_1} \ket{\psi} = (-\Id)^{c_i} \ket{\psi},
 	\end{align*}
 	where the last equality is Equation \eqref{E:observables1}. Part (c) follows
 	similarly from Equation \eqref{E:observables2}.
 \end{proof}
 
 \begin{cor}\label{cor::approxstrategy}
 	Using the notation and hypotheses of Proposition \ref{prop::approxstrategy},
 	if we define $\Phi, \Phi' : \mcF(x_1,\ldots,x_n,J) \arr \mcU(\mcH)$
 	by
 	\begin{equation*}
 	\Phi(x_j) = B_j \text{ for all } 1 \leq j \leq n, \quad
 	\Phi(J) = -\Id 
 	\end{equation*}
 	and
 	\begin{equation*}
 	\Phi'(x_j) = \begin{cases} 
 	A_{ij} & \text{ any } i \text{ such that } j \in V_i \\
 	\Id & \text{ if no such } i \text{ exists}
 	\end{cases}, \quad
 	\Phi'(J) = -\Id
 	\end{equation*}
 	then $(\Phi,\Phi')$ is a $(\tau, \kappa)$-bipartite representation of
 	$\Gamma_{Mx=c}$ with respect to $\ket{\psi}$, where 
 	\begin{equation*} 
 	\tau = 2 \max(\tR,4) \sqrt{\tM (\eps + 2^{\tR-1} \delta)} + \binom{\max(\tR,4)}{2}
 	2^{\tR+1} \delta 
 	\end{equation*}
 	and $\kappa = 2^{\tR+1} \delta$.
 \end{cor}
 \begin{proof}
 	Follows immediately from Proposition \ref{prop::approxstrategy}, 
 	Equation \eqref{E:observables3}, and the fact that $B_j^2=\Id$. 
 \end{proof}
 
 \section{Embedding finitely-presented groups in solution groups}
 
 By Theorem \ref{thm::SBR}, every recursive language can be efficiently encoded
 as the word problem of a finitely-presented group. By Lemma
 \ref{lem::HNNtrick}, the word problem for finitely-presented groups reduces to
 the problem of determining whether $J_G=1$ in finitely-presented groups $G$
 over $\Z_2$. By Theorem \ref{thm::perfectstrats}, if $G = \Gamma_{Mx=c}$ is
 a solution group, then $J_G=1$ if and only if $\omega^{co}(\mcG_{Mx=c})=1$.
 
 The main result of \cite{Sl16} is that the problem of determining whether
 $J_G=1$ for general finitely-presented groups $G$ over $\Z_2$ reduces to the
 problem of determining whether $J_{\Gamma}=1$ for solution groups $\Gamma =
 \Gamma_{Mx=c}$. In this paper, we use the following version of this result:
 \begin{thm}[\cite{Sl16}]\label{thm::solutionembedding}
 	Let $G = \langle S : R \rangle$ be a finitely presented group over $\Z_2$,
 	such that $J_G \in S$, and let $N = |S| + \sum_{r \in R} |r|$ be the size
 	of the presentation. Then there is an $m \times n$ linear system $Mx=c$
 	and a map $\phi : \mcF(S) \arr \mcF(x_1,\ldots,x_n,J)$ such that
 	\begin{enumerate}[(a)]
 		\item $\phi(J_G) = J_{\Gamma}$, and $\phi$ descends to an injection $G
 		\arr \Gamma_{Mx=c}$ (in other words, for all $w \in \mcF(S)$,
 		$\phi(w)$ is trivial in $\Gamma_{Mx=c}$ if and only if $w$ is
 		trivial in $G$);
 		\item for all $w \in \mcF(S)$, $|\phi(w)| \leq 4 |w|$, and if
 		$w$ is trivial in $G$, then $\Ar_{\Gamma}(\phi(w)) = O(N \cdot
 		\Ar_{G}(w))$; and 
 		\item $M$ has exactly three non-zero entries in every row, 
 		the dimensions $m$ and $n$ of $M$ are $O(N)$, and $M$ and $b$ can be constructed from
 		$\langle S : R \rangle$ in time polynomial in $N$. 
 	\end{enumerate}
 \end{thm}
 Note that if $G = \langle S : R\rangle$ and $G' = \langle S' : R' \rangle$ are
 finitely-presented groups, and $\phi : \mcF(S) \arr \mcF(S')$ is a homomorphism
 which descends to a homomorphism $G \arr G'$, then $\Ar_{G'}(\phi(w)) =
 O(\Ar_G(w))$, with a constant which depends on $G$, $G'$, and $\phi$. The
 statement in part (b) of Theorem \ref{thm::solutionembedding} is stronger, in
 that the constant is independent of $G$ (so the only dependence on $G$ comes
 from $N$). 
 
 \begin{proof}[Proof of Theorem \ref{thm::solutionembedding}]
 	Part (a) is Theorem 3.1 of \cite{Sl16}. For the complexity statements in
 	parts (b) and (c), we need to analyze the construction of $M$ and $b$,
 	which occurs in Proposition 4.3, Corollary 4.8, and Theorem 5.1 of
 	\cite{Sl16}. For this purpose, suppose that $G = \langle S : R \rangle$ is
 	a finitely presented group over $\Z_2$. For simplicity, we assume that $J_G
 	= J \in S$, and that all relations containing $J$ are of the form $J\cdot r = 1$
 	for some word $r \in \mcF(S \setminus \{J\})$. This assumption can always
 	be satisfied by adding an extra generator.
 	
 	For the first step of the construction, we also need some notation.  If $x
 	\in \mcF(S')$ is equal to the reduced word $s_1^{a_1} \cdots s_k^{a_k}$,
 	where $s_i \in S'$ and $a_i \in \{\pm 1\}$ for all $1 \leq i \leq k$, let
 	$x^+ = s_1 \cdots s_k$. Note that this word is still reduced, and that $x$
 	and $x^+$ represent the same element in the group 
 	\begin{equation*}
 	\langle S' : s^2 = 1 \text{ for all } s \in S'\rangle.
 	\end{equation*} 
 	Now, starting from $G = \langle S : R \rangle$, we take a new set of
 	indeterminates $S' = \{u_s,v_s : s \in S \setminus \{J\}\}$, and define
 	$\phi_1 : \mcF(S) \arr \mcF(S'\cup\{J\})$ by $\phi_1(s) = u_s v_s u_s v_s$
 	for all $s \in S \setminus \{J\}$ and $\phi_1(J) = J$.  We then let 
 	\begin{equation*}
 	G' = \langle S'\cup\{J\}: R' \cup \{u_s^2 = v_s^2 = 1 : s \in S \setminus \{J\}\} \cup \{J^2 = 1\}
 	\rangle,
 	\end{equation*}
 	where $R' = \{\phi_1(r)^+ : r \in R\}$. Since $u_s^2 = v_s^2 = J^2 = 1$ in
 	$G'$, we conclude that $\phi_1$ descends to a homomorphism $\phi_1 : G \arr
 	G'$. It is not hard to see that this morphism is injective (see, for instance,
 	\cite[Proposition 4.3]{Sl16}), and clearly $|\phi_1(w)| \leq 4 |w|$. If $r
 	\in R$, then $\phi_1(r)$ can be turned into $\phi_1(r)^+$ in at most $4|r|$
 	applications of the relations $u_s^2=v_s^2=1$, $s \in S \setminus \{J\}$.
 	(In particular, $\Ar_{G'}(\phi_1(w)) \leq 4 N \Ar_{G}(w)$, although we use
 	a more refined calculation for bound on $\Ar_{\Gamma}$ in part (b).) The
 	size of the presentation of $G'$ is 
 	\begin{equation*}
 	N' = |S'| + 1 + \sum_{r \in R'} |r| + 4|S|-2 \leq 6|S| + 4 \sum_{r \in R} |r|
 	\leq 6 N,
 	\end{equation*}
 	and the presentation can be constructed from $\langle S : R \rangle$ in
 	$O(N)$ time.
 	
 	To finish the construction of $Mx=c$, we apply the wagon wheel construction
 	from Section 5 of \cite{Sl16} to the group $G'$. This construction is best
 	understood pictorially. An $m \times n$ matrix $M$ with entries in $\Z_2$
 	can be represented graphically by drawing a hypergraph with a vertex for
 	each row of $M$, and an edge for each column, such that the $j$th hyperedge
 	is incident to the $i$th vertex if and only if $M_{ij}=1$. With this
 	representation, a vector $b \in \Z_2^m$ is the same as function from the
 	vertices of the hypergraph to $\Z_2$. So a linear system $Mx=c$ can thus be
 	represented by a hypergraph with a (not necessarily proper) $\Z_2$-vertex
 	colouring, where the edges correspond to the variables of the system, and
 	the vertices to the equations. 
 	
 	In the wagon wheel construction, $Mx=c$ is defined as a union of subsystems
 	$M^r x^r = c^r$, each corresponding to a relation $r \in R'$. The variables
 	of $Mx=c$ consist of the indeterminates $S'$, as well as an additional set
 	of ancillary variables $S''$. Each ancillary variable appears in exactly
 	one of the subsystems $M^r x^r = c^r$, while the variables $S'$ are shared.
 	If $r = J^p s_1 \cdots s_n$, where $p \in \Z_2$ and $s_1,\ldots,s_n \in
 	S'$, then the portion of the hypergraph of $Mx=c$ corresponding to $M^r x^r = c^r$
 	is shown in Figure \ref{fig::wagonwheel}, with the ancillary variables denoted
 	by $a_i,b_i,c_i,d_i$, $1 \leq i \leq n$. The vertex colouring is also shown in
 	Figure \ref{fig::wagonwheel}: one vertex is given colour $p$, and the remaining
 	vertices are coloured $0$.
 	
 	\begin{figure}
 		\begin{tikzpicture}[auto,ultra thick,scale=.6,emptynode/.style={inner sep=0},
    every node/.style={scale=.8},
    vertex/.style={circle,draw,thin,inner sep=2.5,scale=.55,fill=white},
    helabel/.style={fill=white,scale=.8}]
    \draw (-10:3) arc [start angle=-10,end angle=150,radius=3];
    \draw[dashed] (150:3) arc [start angle=150,end angle=350,radius=3];
    \node[vertex] (0) at (-10:3) {$0$};
    \node[vertex] (1) at (30:3) {$0$};
    \node[vertex] (2) at (70:3) {$0$};
    \node[vertex] (3) at (110:3) {$0$};
    \node[vertex] (4) at (150:3) {$0$};
    \path (0) arc[start angle=-10,end angle=30,radius=3] node[swap,pos=.35] {$d_3$};
    \path (1) arc[start angle=30,end angle=70,radius=3] node[swap,pos=.5] {$d_2$};
    \path (2) arc[start angle=70,end angle=110,radius=3] node[swap,pos=.5] {$d_1$};
    \path (3) arc[start angle=110,end angle=150,radius=3] node[swap,pos=.5] {$d_{n}$};

    \draw (-10:6) arc [start angle=-10,end angle=150,radius=6];
    \draw[dashed] (150:6) arc [start angle=150,end angle=350,radius=6];
    \node[vertex] (5) at (30:6) {$0$}
        edge node[swap,pos=.3] {$c_{2}$} (1);
    \node[vertex] (6) at (50:6) {$0$};
    \node[vertex] (7) at (70:6) {$0$}
        edge node[swap] {$c_{1}$} (2);
    \node[vertex] (8) at (90:6) {$p$};
    \node[vertex] (9) at (110:6) {$0$}
        edge node[swap,pos=.3] {$c_{n}$} (3);
    \node[vertex] (10) at (130:6) {$0$};
    \node[vertex] (11) at (150:6) {$0$}
        edge node[swap] {$c_{n-1}$} (4);
    \node[vertex] (12) at (10:6) {$0$};
    \node[vertex] (13) at (-10:6) {$0$}
        edge node[swap] {$c_3$} (0);
    \path (5) arc[start angle=30, end angle=50,radius=6] node[pos=.4,swap] {$b_2$};
    \path (6) arc[start angle=50, end angle=70,radius=6] node[pos=.58,swap] {$a_2$};
    \path (7) arc[start angle=70, end angle=90,radius=6] node[pos=.58,swap] {$b_1$};
    \path (8) arc[start angle=90, end angle=110,radius=6] node[pos=.32,swap] {$a_1$};
    \path (9) arc[start angle=110, end angle=130,radius=6] node[pos=.35,swap] {$b_{n}$};
    \path (10) arc[start angle=130, end angle=150,radius=6] node[pos=.5,swap] {$a_{n}$};
    \path (13) arc[start angle=-10,end angle=10,radius=6] node[pos=.5,swap] {$b_{3}$};
    \path (12) arc[start angle=10,end angle=30,radius=6] node[pos=.5,swap] {$a_{3}$};

    \path[fill,pattern=north west lines] (12.20) to ++(40:2) to ($(12.0)+(-20:2)$) 
        to (12.0) to [bend right] (12.20);
    \draw[thin] (12.20) to ++(40:2);
    \draw[thin] (12.0) to node[helabel,pos=.35,yshift=4] {$s_{3}$} ++(-20:2);

    \path[fill,pattern=north west lines] (6.60) to ++(80:2) to ($(6.40)+(20:2)$) 
        to (6.40) to [bend right] (6.60);
    \draw[thin] (6.60) to ++(80:2);
    \draw[thin] (6.40) to node[helabel,pos=.55,yshift=2] {$s_{2}$} ++(20:2);

    \path[fill,pattern=north west lines] (8.100) to ++(120:2) to ($(8.80)+(60:2)$) 
        to (8.80) to [bend right] (8.100);
    \draw[thin] (8.100) to ++(120:2);
    \draw[thin] (8.80) to node[helabel,pos=.5,xshift=-5] {$s_{1}$} ++(60:2);

    \path[fill,pattern=north west lines] (10.140) to ++(160:2) to ($(10.120)+(100:2)$) 
        to (10.120) to [bend right] (10.140);
    \draw[thin] (10.140) to ++(160:2);
    \draw[thin] (10.120) to node[helabel,pos=.5,xshift=-5] {$s_{n}$} ++(100:2);
\end{tikzpicture}
 		\caption{Pictorial depiction of the linear system associated to each
 			relation in the wagon wheel embedding as described in the proof of Theorem \ref{thm::solutionembedding}. Figure reproduced from \cite[Figure 2]{Sl16}.}
 		\label{fig::wagonwheel}
 	\end{figure}
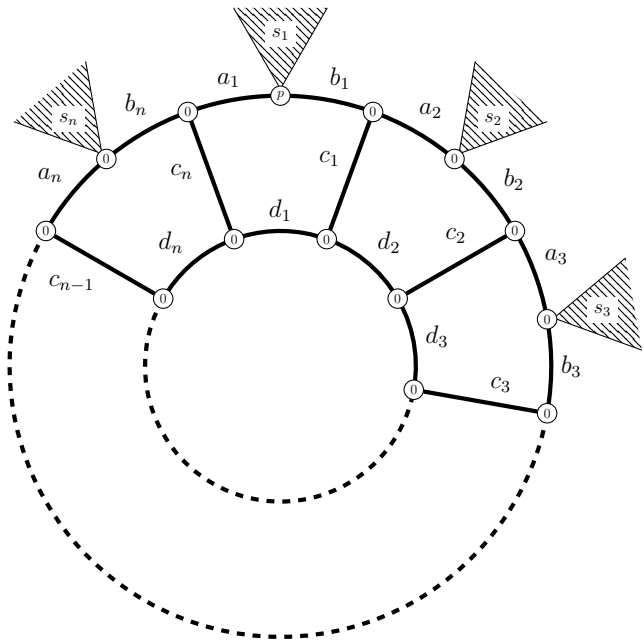
 	
 	As can be seen from Figure \ref{fig::wagonwheel}, 
 	the number of ancillary variables added for subsystem $M^r x^r = c^r$ is
 	$4|r|$, and the number of equations added is $3|r|$. Since every vertex in
 	the hypergraph has degree three, every row of $M^r$ has exactly three
 	non-zero entries.
 	Theorem 5.1 of \cite{Sl16} then states that the natural inclusion $\phi_2 :
 	\mcF(S' \cup \{J\}) \arr \mcF(S' \cup \{J\} \cup S'') : s \mapsto s$
 	descends to an injection $G' \arr \Gamma_{Mx=c}$. 
 	
 	Recall from Definition \ref{def::solutiongroup} that every linear equation
 	in $Mx=c$ becomes a defining relation of $\Gamma := \Gamma_{Mx=c}$. The
 	wagon wheel construction is designed so that if $r \in R'$, then
 	$\phi_2(r)$ can be turned into the identity by applying each defining
    relation from $M_r x = c_r$ exactly once, so $\Ar_{\Gamma}(\phi_2(r)) \leq
    3|r|$ for all $r \in R'$. This is easiest to see using pictures of the
    group, for which we refer to Section 7 of \cite{Sl16}; with this formalism,
    Figure \ref{fig::wagonwheel} is itself a proof that $\phi_2(r)=1$, with
    each vertex corresponding to a use of the corresponding relation. For
    relations $r = J^p s_1 \cdots s_n$ with $p \neq 1$, we start with the 
    relation coloured by $p$, after which $J$ no longer appears in the word.
    If $r \in R$, then $\phi_2(\phi_1(r))$ can be turned
 	into the identity with at most $7|r|$ applications of the relations of
 	$\Gamma$, by first changing $\phi_2(\phi_1(r))$ to $\phi_2(\phi_1(r)^+)$
 	using the relations $s^2=1$, $s \in S'$, and then applying the linear
 	relations of $\Gamma$. It follows that $\Ar_{\Gamma}(\phi_2(w)) = O(N
 	\Ar_{G}(w))$ for all $w \in \mcF(S)$ which are trivial in
 	$G$. It should also be clear from Figure \ref{fig::wagonwheel} that $M^r
 	x^r = c^r$ can be constructed in time polynomial in $|r|$. We conclude that
 	$Mx=c$ is an $m \times n$ linear system with $m$ and $n$ equal to $O(N)$,
 	and that $M$ and $b$ can be constructed in time polynomial in $N$, so the
 	theorem holds with $\phi = \phi_2 \circ \phi_1$.
 \end{proof}
 
 Theorem \ref{thm::solutionembedding} is sufficient to prove Theorem \ref{T:main_game}.
 However, to get perfect zero-knowledge protocols for $\MIP^{co}_{\delta}$, we 
 need to prove an additional fact about the embedding in Theorem
 \ref{thm::solutionembedding}.
 
 \begin{lemma}\label{L:pzkwagonwheel}
 	Let $Mx=c$ be an $m \times n$ linear system from the wagon wheel
 	construction in the proof of Theorem \ref{thm::solutionembedding}. In the
 	solution group $\Gamma_{Mx=c}$, the generator $x_i$ is not equal to $1$ or
 	$J$ for all $1 \leq i \leq n$, and similarly the product $x_i x_j$ is not
 	equal to $1$ or $J$ for all $1 \leq i \neq j \leq n$.
 \end{lemma}
 \begin{proof}
 	We revisit the wagon wheel construction in the proof of Theorem
 	\ref{thm::solutionembedding}. We need to show that $x_i \neq 1$ and $x_i \neq x_j$
 	in $\Gamma_0 := \Gamma_{Mx=c} / \langle J \rangle$ for all $1 \leq i \neq j
 	\leq n$. This is the same as showing that $x_i \neq 1$
 	and $x_i \neq x_j$ in $\Gamma_{Mx=0} = \Gamma_0 \times \Z_2$.
 	Recall that the generators of $\Gamma_{Mx=0}$ are split into two sets, the
 	generators $S'$ of $G'$, and the ancillary variables $S''$. The group $G'_0
 	:= G' / \langle J \rangle$ has a presentation where every generator $s \in
 	S'$ occurs an even number of times in every relation. Thus for any $s\in S'$,
 	we can define a representation $G' \arr \C^{x}$ by sending $s \in S'$ to
 	$-1$, and $t \in S' \setminus \{s\}$ to $1$. It follows that $s \neq 1$
 	and $s \neq t$ in $G'_0$ for every $s \neq t \in S'$. Since $G_0' \arr \Gamma_0$
 	is an injection, we conclude that the same holds in $\Gamma_0$.
 	
 	For the ancillary variables, consider the hypergraph description of the
 	system $Mx=0$. Given a subset of edges $C$, let $y \in \Z_2^n$ be the
 	vector with $y_i = 1$ if and only if the $i$th edge is in $C$. Then $y$ is
 	a classical solution to $Mx=0$ if and only if every vertex of the
 	hypergraph is incident with an even number of edges from $C$. The classical
 	solutions of $Mx=0$ correspond to $1$-dimensional representations of
 	$\Gamma_0$; if $y$ is a solution of $Mx=0$, then the corresponding
 	$1$-dimensional representation of $\Gamma_0$ sends $x_i \mapsto (-1)^{y_i}$. 
 	
 	Inspecting the wagon wheel hypergraph in Figure \ref{fig::wagonwheel}, we
 	see that every ancilla variable $s \in S''$ belongs to a cycle $C$ which
 	does not contain any edges from $S'$. Using the corresponding
 	representation of $\Gamma_0$, we see that $s\neq 1$ and $s \neq t$ in
 	$\Gamma_0$ for all $s \in S''$ and $t \in S'$. Similarly, if $s \neq t \in
 	S''$, and $\{s,t\}$ is not one of the pairs $\{a_i,b_i\}$, then there is a
 	cycle $C$ containing $s$ and not containing $t$, so $s \neq t$ in $\Gamma_0$.
 	
 	For the pairs $\{a_i,b_i\}$, fix $s \in S''$, and recall that if $r = s_1
 	\cdots s_n$ is a relation of $G'$, where $s_1,\ldots,s_n \in S'$, then
 	$s$ occurs an even number of times in $r$. Let $1 \leq i_1 < \cdots <
 	i_{2k} \leq n$ be the indices such that $s_{i_j} = s$, and let
 	\begin{align*}
 	C_r := \{& s_{i_1},b_{i_1},a_{i_1+1},b_{i_1+1},\ldots,a_{i_2},s_{i_2},\\
 	& s_{i_3},b_{i_3},a_{i_3+1},b_{i_3+1},\ldots,a_{i_4},s_{i_4},\\
 	& \ldots,\\
 	& s_{i_{2k-1}},b_{i_{2k-1}},a_{i_{2k-1}+1},b_{i_{2k-1}+1},\ldots,
 	a_{i_{2k}},s_{i_{2k}} \}.
 	\end{align*}
 	be the collection of paths along the outer cycle of the wagon wheel
 	connection $s_{i_1}$ with $s_{i_2}$, $s_{i_3}$ with $s_{i_4}$, and so on.
 	Let $C := \bigcup_{r \in R'} C_r$. Then every vertex of the hypergraph of
 	$Mx=0$ is incident to an even number of edges in $C$.  If we look at a
 	particular relation $r$, then for every $1 \leq j \leq 2k$, exactly one of
 	the edges $a_{i_j},b_{i_j}$ belongs to $C_r$, so $a_{i_j} \neq b_{i_j}$ in $\Gamma_0$. 
 	It follows that all of the pairs of ancillary generators $a_i,b_i$ are
 	distinct in $\Gamma_0$.
 \end{proof} 
 
 \begin{prop}\label{P:pzkwagonwheel}
 	Let $Mx=c$ be an $m \times n$ linear system from the wagon wheel
 	construction in the proof of Theorem \ref{thm::solutionembedding}, 
 	and suppose $J \neq 1$ in $\Gamma_{Mx=c}$. Then $\mcG_{Mx=c}$ has
 	a commuting-operator strategy $\mcS = (\mcH,\{P^i_a\}_{a \in \mcO^i_a},
 	\{Q^j_b\}_{b \in \Z_2}, \ket{\psi})$ such that $\omega(\mcG_{Mx=c};\mcS) =
 	1$, and 
 	\begin{equation*}
 	\bra{\psi} P^i_a Q^j_b \ket{\psi} = \begin{cases} 
 	\frac{1 + (-1)^{a_j+b}}{8} & j \in V_i \\
 	\frac{1}{8} & j \not\in V_i
 	\end{cases}
 	\end{equation*}
 	for all $1 \leq i \leq m$, $a \in \mcO_A^i$, $1 \leq j \leq m$, $b \in \Z_2$.
 \end{prop}
 \begin{proof}
 	Suppose $J \neq 1$ in $\Gamma_{Mx=c}$. We recall the construction of a perfect
 	commuting-operator strategy for $\mcG_{Mx=c}$ from \cite{CLS16}. Let $\mcH
 	= \ell^2 \Gamma_{Mx=c}$ be the regular representation of $\Gamma_{Mx=c}$, and
 	given $g \in \Gamma_{Mx=c}$, let $L(g)$ (resp. $R(g)$) denote left (resp. right)
 	multiplication by $g$. Then $L(g)$ and $R(g)$ are unitaries for all $g \in \Gamma_{Mx=c}$,
 	and we can get a perfect strategy for $\mcG_{Mx=c}$ by taking $A_{ij} = L(X_j)$
 	for all $1 \leq i \leq m$, $j \in V_i$, $B_j = R(X_j)$ for all $1 \leq j \leq n$, and
 	$\ket{\psi} = \frac{1-J}{\sqrt{2}}$ considered as an element of $\mcH$. 
 	Since $J$ is central of order $2$, we have that
 	\begin{equation*}
 	\bra{\psi} A_{ik} B_j \ket{\psi} = \bra{\psi} L(X_k)R(X_j) \ket{\psi}
 	= \bra{\psi} R(X_k X_j) \ket{\psi} = \begin{cases} 1 & X_k X_j = 1 \\
 	-1 & X_k X_j = J \\
 	0 & \text{otherwise}.
 	\end{cases}.
 	\end{equation*}
    Recall from Equation \eqref{E:observables4} that
 	\begin{equation*}
 	P_a^i = \prod_{k \in V_i} \left(\frac{\Id + (-1)^{a_k} A_{ik}}{2}\right)
 	\end{equation*}
    for all $a \in \mcO_A^i$ and $Q^j_b = \frac{1 + (-1)^b B_j}{2}$ for all
    $b \in \mcO_B^j$.
    Using the fact that $\prod_{k \in V_i} A_{ik} = (-\Id)^{c_i}$ in perfect
    strategies,  and that $|V_i|=3$ in the linear system constructed in Theorem
    \ref{thm::solutionembedding}, we get that
 	\begin{equation*}
 	P_a^i = \prod_{k \in V_i} \left(\frac{\Id + (-1)^{a_k} A_{ik}}{2}\right)
 	= \frac{\Id}{4} + \frac{1}{4} \sum_{k \in V_i} (-1)^{a_k} A_{ik}.
 	\end{equation*}
    By Lemma \ref{L:pzkwagonwheel}
 	\begin{equation*}
 	\bra{\psi} P^i_a Q^j_b \ket{\psi} = \frac{1}{8} + \frac{1}{8} \bra{\psi} \sum_{k \in V_i} (-1)^{a_k + b} A_{ik} B_j \ket{\psi}
 	= \begin{cases} \frac{1}{8} & j \not\in V_i \\  
 	\frac{1 + (-1)^{a_j + b}}{8} & j \in V_i.
 	\end{cases}.
 	\end{equation*}
 \end{proof}
 
 \section{Proof of Theorem \ref{T:main_game}}\label{S:mainproof}
 
 In this section we prove Theorem \ref{T:main_game}, by proving the main technical
 result of the paper. 
 \begin{thm}\label{T:gameapprox}
 	Let $L \subset A^*$ be a language over a finite alphabet $A$, and contained in
 	$\NTIME(T(n))$, where $T(n)^4$ is superadditive. Then for any string $w \in A^*$,
 	there is a non-local game $\mcG_w$ such that
 	\begin{enumerate}[(a)]
 		\item the game $\mcG_w$ has question sets of size $O(|w|)$ and output sets
 		of size at most $8$, 
 		\item the function $w \mapsto \mcG_w$ is computable in $O(|w|^k)$-time, where
 		$k$ is some universal constant,
 		\item if $w \not\in L$ then $\omega^{co}(\mcG_w) = 1$, and
 		\item if $w \in L$ then
 		\begin{equation*}
 		\omega^{co}_{\delta}(\mcG_w) \leq 1 -
 		\frac{1}{T(O(|w|))^{k'}} + O\left(\delta\right)
 		\end{equation*}
 		for some universal constant $k'$.
 	\end{enumerate}
 \end{thm}
 While the constants $k,k'$ in Theorem \ref{T:gameapprox} are independent of $L$, 
 the other constants appearing in the big-$O$ can depend on
 $L$. The game $\mcG_w$ will be a linear system game $\mcG_{M(w)x=c(w)}$, where
 $M(w)x =c(w)$ is an $O(|w|) \times O(|w|)$-linear system. Since the linear
 system game of an $m \times n$ linear system $Mx=c$ can be constructed in
 $O(mn)$-time from $M$ and $b$, the goal in proving Theorem \ref{T:gameapprox}
 will be to show that the linear system $M(w)x=c(w)$ can be constructed in time
 polynomial in $|w|$. Theorem \ref{T:main_game} is an immediate corollary of
 Theorem \ref{T:gameapprox}.

 \begin{proof}[Proof of Theorem \ref{T:gameapprox}]
 	Given the language $L$, let $G = \langle S : R \rangle$ be the group from
 	Theorem \ref{thm::SBR}, and let $\kappa$ be the function $A^* \arr
 	\mcF(S)$. Given $w \in A^*$, we let $\widetilde{G}_{\kappa(w)}$ be the
 	group over $\Z_2$ constructed in Definition \ref{def::HNNtrick}, and
 	$M(w)x=c(w)$ be the linear system constructed from
 	$\widetilde{G}_{\kappa(w)}$ in Theorem \ref{thm::solutionembedding}.
 	Finally, we let $\mcG_w := \mcG_{M(w)x=c(w)}$ and $\Gamma_w := 
 	\Gamma_{M(w)x=c(w)}$. The only part of the presentation of
 	$\widetilde{G}_{\kappa(w)}$ that changes with $w$ is the relation
 	$[t,[x,\kappa(w)]]=1$, so the presentation of $\widetilde{G}_{\kappa(w)}$
 	has size $O(|\kappa(w)|) = O(|w|)$, and $M(w)x=c(w)$ is an $O(|w|) \times
 	O(|w|)$ linear system.  Because $M(w)$ has only three non-zero
 	entries per equation, Alice's output sets in $\mcG_w$ will have size $2^3 =
 	8$, while Bob's output sets will have size $2$. Thus parts (a) and (b) of
 	Theorem \ref{T:gameapprox} follow from part (c) of Theorem
 	\ref{thm::solutionembedding}. 
 	
 	By Theorem \ref{thm::SBR} and Lemma
 	\ref{lem::HNNtrick}, if $w \not\in L$ then $\kappa(w) \neq 1$ in $G$, and
 	hence $J \neq 1$ in $\widetilde{G}_{\kappa(w)}$. Since the inclusion
 	$\widetilde{G}_{\kappa(w)} \incl \Gamma_w$ sends
 	$J_{\widetilde{G}_{\kappa(w)}} \mapsto J_{\Gamma_w}$, we conclude that $J
 	\neq 1$ in $\Gamma_w$. By Theorem \ref{thm::perfectstrats},
 	$\omega^{co}(\mcG_w)=1$, proving part (c). 
 	
 	This leaves part (d). Suppose $w \in L$. Then $\kappa(w)=1$ in $G$, $J= 1$
 	in $\widetilde{G}_{\kappa(w)}$, and hence $J=1$ in $\Gamma_w$. Suppose
 	$\mcS$ is a $\delta$-AC operator strategy for $\mcG_w$ with 
 	$\omega(\mcG_w; \mcS) \geq 1-\eps$. Since $M$ has only three non-zero entries
 	per row, the parameters $r$ and $K$ appearing in Corollary \ref{cor::approxstrategy}
 	are $O(1)$ and $O(|w|)$ respectively. Also, because we are interested in
 	$\delta \leq 2$, we can say that $\delta = O(\sqrt{\delta})$. Thus Corollary
 	\ref{cor::approxstrategy} states that there is a $(O(\sqrt{|w|(\eps+\delta)}),
 	O(\delta))$-bipartite representation $(\Phi,\Phi')$ of $\mcG_w$ with
 	respect to the state $\ket{\psi}$ used in $\mcS$. By construction, this
 	bipartite representation has $\Phi(J) = -\Id$. The length of the longest
 	relation in $\Gamma_w$ is $4$, and the length of $J$ in $\Gamma_w$ is $1$,
 	so Lemma \ref{lem::bipartiterep} implies that
 	\begin{equation}\label{E:contradiction}
 	-\ket{\psi} = \Phi(J) \ket{\psi} \approx_{O\left(\Ar_{\Gamma_w}(J)^2\sqrt{|w|(\eps+\delta)}\right)}
 	\ket{\psi}.
 	\end{equation}
 	By Theorem \ref{thm::solutionembedding}, part (b) and Lemma
 	\ref{lem::HNNtrick}, part (c), 
 	\begin{equation*}
 	\Ar_{\Gamma_w}(J) = O\left(|w| \cdot \Ar_{\widetilde{G}_{\kappa(w)}}(J)\right)
 	= O\left(|w| \cdot \Ar_G(\kappa(w))\right).
 	\end{equation*}
 	Finally, by Theorem \ref{thm::SBR}, $|\kappa(u)| = O(|u|)$ and $\Dehn_G$ is
 	bounded by a function equivalent to $T(n)^4$, so there is a constant $C$ such
 	that $\Ar_G(\kappa(w)) = O(T(C|w|)^4 + |w|)$. Since $T(n)^4$ is
 	superadditive by assumption, $|w| = O(T(|w|)^4)$, and we can conclude that
 	$\Ar_{\Gamma_w}(J) = O(T(C|w|)^{8})$. Returning to Equation \eqref{E:contradiction},
 	since $\norm{-\ket{\psi} - \ket{\psi}} = 2$, we see that there is a
 	constant $C_0>0$ such that
 	\begin{equation*}
 	C_0 T(C|w|)^{18}\sqrt{\eps + \delta} \geq 2.
 	\end{equation*}
 	Hence 
 	\begin{equation*}
 	C_0^2 T(C|w|)^{36}(\eps + \delta) \geq 4.
 	\end{equation*}
 	so,
 	\begin{equation*}
 	\eps \geq \frac{4}{C_0^2 T(C|w|)^{36}} - \delta,
 	\end{equation*}
 	So we conclude that
 	\begin{equation*}
 	\omega(\mcG_w;\mcS) \leq 1 - \Omega\left(\frac{1}{T(C|w|)^{36}}\right)
 	+ O(\delta).
 	\end{equation*}
 	Because $T(n)^4$ is superadditive, $T(C_0 \cdot C|w|)^{36} \geq C_0 T(C|w|)^{36}$
 	for any integer $C_0$, so we can move the constant from the big-$\Omega$ inside
 	$T$, proving part (d). 
 \end{proof}

 \section{Multi-prover interactive proofs}
 
 In this section we define the complexity class
 $\PZKMIP^{co}_{\delta}(2,1,1,1-1/f(n))$, and prove Theorem \ref{T:main_MIP}.
 We first recall the definition of $\MIP^{co}_{\delta}$. The definition given
 here is a simple variant on Definition 8 of \cite{CV15}.
 \begin{defn}\label{defn::mip}
 	A language $L$ over an alphabet $A$ is in the class
 	$\MIP^{co}_{\delta}(2,1,1,1-1/f(n))$ of multi-prover interactive proofs
 	with two provers, one round, completeness probability $1$, and soundness probability
 	$1 - 1/f(n)$, if and only if there is family of two-player non-local games
 	$\mcG_w = (\mcI_A^w,\mcI_B^w,\mcO_A^{*,w}, \mcO_B^{*,w}, V_w, \pi_w)$ indexed by strings $w
 	\in A^*$, such that 
 	\begin{itemize}
 		\item the input sets $\mcI_A^w$, $\mcI_B^w$ and output sets $\mcO^{*,w}_A$,
 		$\mcO^{*,w}_B$ for $\mcG_w$ are subsets of strings of length $\poly |w|$
 		(and hence can have size at most $2^{\poly |w|}$). 
 		\item the function $V_w$ can be computed in polynomial time in $|w|$
 		and the lengths of its inputs,
 		\item the distribution $\pi_w$ can be sampled in polynomial time in
 		$|w|$ and the lengths of its inputs, 
 		\item (completeness) if $w \in L$ then $\omega^{co}(\mcG_w) = 1$, and
 		\item (soundness) if $w \not\in L$ then $\omega^{co}_{\delta}(\mcG_w) \leq 1 - 1/f(|w|)$.
 	\end{itemize}
 	The family $\{\mcG_w\}$ is referred to as a \emph{protocol} for $L$. 
 \end{defn} 
 Here $\delta$ can also be a function of $|w|$. When $\delta=0$,
 $\MIP^{co}_0$ is the class of commuting-operator multi-prover interactive
 proofs, which dates back to \cite{IKPSY08}. Note that, in Definition \ref{defn::mip}, the
 protocol must be sound against $\delta$-AC operator strategies, whereas the completeness condition requires a perfect commuting-operator strategy. As a result,
 $\MIP^{co}_{\delta} \subset \MIP^{co}$ for all $\delta$. 
 
 \begin{rmk}\label{rmk::povm}
 	Our definition is slightly different from \cite{CV15} in that we
 	use $\delta$-AC strategies with projective measurements, rather than POVMs.
 	It's not clear how this changes the complexity class in general, since we
 	are restricting the class of strategies that a protocol must be sound against
 	(which potentially strengthens the class) and restricting the class of strategies
 	that can be used for completeness (which potentially weakens the class). 
 	However, Claim 9 of \cite{CV15} shows that projective measurements and POVMs
 	are equivalent up to an increase in $\delta$ proportional to the size of the
 	output sets. Since our lower bounds use protocols with a constant number of
 	outputs, the lower bounds will also apply if we define $\MIP^{co}_{\delta}$
 	using POVMs. 
 \end{rmk}
 
 
 Next we will define the perfect zero knowledge version of $\MIP^{co}_{\delta}$,
 called $\PZKMIP^{co}_{\delta}$.   Informally, a multi-prover interactive proof
 is perfect zero-knowledge if the verifier gains no new information from
 interacting with the provers. This is formalized by requiring that, for every
 yes instance, the provers have a strategy for which the verifier can
 efficiently simulate the provers' behaviour. 
 
 Let $\mcG = (\mcI_A,\mcI_B,\mcO_A^*, \mcO_B^*, V, \pi)$ be a non-local game.
 If the players use a commuting-operator strategy given by measurements
 $\{P^x_a\}$ and $\{Q^y_b\}$ and a state $\ket{\psi}$ in a Hilbert space $\mcH$,
 then to an outside party (such as the verifier), the players actions are
 described by the probabilities 
 \begin{equation*}
 p(a,b|x,y) = \bra{\psi} P^x_a Q^y_b \ket{\psi}.
 \end{equation*}
 When $x,y$ are fixed, $p(a,b|x,y)$ gives a probability distribution over
 outcomes $(a,b) \in \mcO_A^x \times \mcO_B^y$.  The family of probability
 distributions $\underline{p} = \{p(a,b|x,y) : (x,y) \in \mcI_A \times \mcI_B,
 (a,b) \in \mcO_A^x \times \mcO_B^y\}$ is called the \emph{correlation matrix}
 of the strategy.
 
 In a interactive proof system, a record of interactions between verifier and
 provers is called a transcript. Let $\{\mcG_w\}$ be a
 $\MIP^{co}_{\delta}(2,1,1,s)$ protocol for a language $L$ as in Definition
 \ref{defn::mip}. During the game $\mcG_w$, the transcript consists simply of
 the inputs $(x,y) \in \mcI^w_A \times \mcI^w_B$ sent to the provers, and the
 outputs $(a,b) \in \mcO_A^x \times \mcO_B^y$ received back. If the verifier
 asks questions $x,y$ with probability $\pi(x,y)$, then the distribution over
 transcripts $(x,y,a,b)$ is given by $\pi(x,y) p(a,b|x,y)$, where
 $\{p(a,b|x,y)\}$ is the correlation matrix of the provers' strategy. A strategy
 is said to be \emph{perfect zero-knowledge against an honest verifier} if it is
 possible to sample from the distribution $\{\pi_w(x,y)
 p(a,b|x,y)\}_{(x,y,a,b)}$ in polynomial time. However, this assumes that the
 verifier chooses questions $x,y$ according to the probability distribution
 $\pi_w$ given in the protocol, something that the provers cannot validate
 themselves while the game is in progress. To be perfect zero-knowledge against
 a possibly dishonest verifier, it is necessary that the verifier be able to
 simulate $\pi(x,y) p(a,b|x,y)$ for any (simulable) distribution $\pi(x,y)$ on
 inputs. This is equivalent to being able to simulate the distributions
 $\{p(a,b|x,y)\}$, so we make the following definition:
 \begin{defn}\label{defn::pzkmip}
 	Let $\{\mcG_w\}$ be a $\MIP^{co}_{\delta}(2,1,1,1-s)$-protocol for a
 	language $L$. Then $\{\mcG_w\}$ is said to be \emph{perfect zero-knowledge}
 	if for each string $w$ and pair $(x,y) \in \mcI_A \times \mcI_B$, there is
 	a probability distribution $\{p_w(a,b|x,y) : (a,b) \in \mcO_A^x \times
 	\mcO_B^y \}$ over $\mcO_A^x \times
 	\mcO_B^y$ such that
 	\begin{enumerate}
 		\item the distribution $\{p_w(a,b|x,y)\}$ can be sampled in polynomial
 		time in $|w|$, $|x|$, and $|y|$, and
 		\item if $w \in L$, then $\{p_w(a,b|x,y) : (x,y) \in \mcI_A \times \mcI_B,
 		(a,b) \in \mcO_A^x \times \mcO_B^y \}$ is the correlation matrix of
 		a commuting-operator strategy $\mcS$ with winning probability
 		$\omega(\mcG_w; \mcS) = 1$. 
 	\end{enumerate}
 	The class $\PZKMIP^{co}_{\delta}(2,1,1,1-1/f(n))$ is the class of languages
 	in $\MIP^{co}_{\delta}(2,1,1,1-1/f(n))$ with a perfect zero-knowledge
 	protocol.
 \end{defn} 
 
 \begin{proof}[Proof of Theorem \ref{T:main_MIP}]
 	Theorem \ref{T:gameapprox} immediately implies that any language $L \in
 	\coNTIME(f(n))$ has a protocol in $\MIP^{co}_{\delta}(2,1,1,1-1/f(Cn)^k)$
 	for some constants $C$ and $k$, where $\delta = o(1/f(Cn)^{k})$. Since the games constructed in the proof
 	of Theorem \ref{T:gameapprox} come from the wagon wheel construction,
 	Proposition \ref{P:pzkwagonwheel} implies that when $w \in L$, the game
 	$\mcG_w$ has a perfect commuting operator strategy with a correlation
 	that can easily be simulated by the verifier. 
 \end{proof}
 
 \subsection{Upper bounds}
 
As mentioned in the introduction, no upper bound on $\MIP^{co}$ is known, but
an upper bound on $\MIP^{co}_{\delta}$ follows from \cite{CV15} as we will now
describe. Recall that $\omega^*_{\delta}(G)$ is the supremum of winning probabilities
across finite-dimensional $\delta$-AC strategies. 
\begin{theorem}[\cite{CV15}, Theorem 2 and Claim 9]\label{thm:sdp}
    Let $G$ be a $2$-prover non-local game in which each prover has $\ell$
    possible answers, and $\omega_{QCSDP}^N(G)$ be the optimum of the $N$-th
    level of the QC SDP hierarchy for $G$. Then $\omega^*_{\delta}(G)
    \geq \omega_{QCSDP}^N(G)$ for all $\delta \geq \Omega(\ell^4/\sqrt{N})$.
\end{theorem}
The original statement of \cite[Theorem 2]{CV15} assumes POVM strategies rather
than projective strategies, for which it is possible to take
$\delta\geq\Omega(\ell^2/\sqrt{N})$. By \cite[Claim 9]{CV15}, it is possible to
round POVM strategies achieving $\omega_{QCSDP}^N(G)$ to projective strategies
by adding a factor of $\ell^2$, leading to $\delta\geq\Omega(\ell^4/\sqrt{N})$.
The QC SDP hierarchy for a non-local game $G$ is as in Definition 10 of
\cite{CV15}.  For our purposes the only properties of the QC SDP hierarchy that
we will require are the following: 
 
\begin{fact} \label{fct:upperbound}
    $\omega_{QCSDP}^N(G)$ is a non-increasing function of $N$, and
    $\omega_{QCSDP}^N(G) \geq \omega^{co}(G)$ for all $N$.  This is an
    elementary property of the hierarchy and is discussed in \cite{CV15}.
\end{fact}
 
\begin{fact}\label{fct:sdpcomp}
    The quantity $\omega_{QCSDP}^N(G)$ can be computed to additive precision $\eps$ 
    in time polynomial in
    $(Q\ell)^{N}$ and $\log(1/\eps)$, where $Q$ is the maximum number of questions to either prover in
    $G$, and $\ell$ is the maximum number of answers.  This is because
    $\omega_{QCSDP}^N(G)$ is defined (in Definition 10 of \cite{CV15}) to be the
    optimal value of an semi-definite program on $\poly((Q\ell)^{N})$ dimensional
    space, with $\poly((Q\ell)^{N})$ constraints.
\end{fact}
 
Suppose that one wishes to decide whether a non-local game $G$ has
$\omega^{co}(G) = 1$, or has  $\omega^{co}_{\delta}(G) \leq 1 - 1/f$ promised
that one of the two is the case.  By Theorem 2 of \cite{CV15}, there exists $M
= O(\ell^8/\delta^2)$ such that $\omega^{co}_{\delta}(G) \geq
\omega_{QCSDP}^{M}(G)$.  If $\omega^{co}(G) = 1$ then $\omega_{QCSDP}^M(G) \geq
\omega^{co}(G) = 1$ by Fact \ref{fct:upperbound}. On the other hand, if
$\omega^{co}_{\delta}(G) \leq 1 - 1/f$ then $\omega^{M}_{QCSDP}(G) \leq 1-1/f$.
It follows by Fact \ref{fct:sdpcomp} that this decision problem can be solved
in time that is polynomial in $(Q\ell)^{M} = (Q\ell)^{O(\ell^8/\delta^2)}$,
where $Q$ and $\ell$ are the sizes of the question and answer sets in $G$
respectively (this argument is made in \cite{CV15} for $\omega^*$ and
$\omega^*_{\delta}$, but it works equally well for $\omega^{co}$ and
$\omega^{co}_{\delta}$). To summarize, we have the following theorem:
\begin{thm}[\cite{CV15}]\label{T:upperbound}
    Let $\delta = o(1/f(n))$, and suppose $L \in \MIP^{co}_{\delta}(2,1,1,1-1/f(n))$. 
    If $f(n)$ is at least exponential, or $L$ has a protocol with constant size
    output sets and $f(n)$ is at least polynomial, then $L \in 
    \TIME(\exp(1 / \poly(\delta)))$. 
\end{thm}
\begin{proof}
    A protocol for $L$ consists of a family of games $\mcG_w$ with question and
    answer sets of size $2^{\poly(n)}$, where $n = |w|$. Suppose that there is
    a constant $C$ such that $f(n) \geq C^n$ for all $n \geq 1$, so that
    $\exp(\poly(n)) \leq \poly(f(n)) \leq 1/\poly(\delta)$. Setting $Q = \ell =
    \exp(\poly(n))$ in the discussion above, we see that there is an algorithm to
    decide whether $\omega^{co}(\mcG_w) =1$ or $\omega^{co}_{\delta}(\mcG_w)
    \leq 1 - 1/f(n)$ (and hence whether $w \in L$) with running time at most
    \begin{equation*}
        \exp(\poly(n) \cdot \exp(\poly(n)) / \delta^2) \leq \exp(1/ \poly(\delta)). 
    \end{equation*}
    
    Similarly, if $L$ has a protocol with constant size output sets, then setting
    $Q = \exp(\poly(n))$ and $\ell = O(1)$, we see that there is an algorithm
    to decide whether $w \in L$ with running time at most
    $\exp(\poly(n) / \delta^2) \leq \exp(1/\poly(\delta))$. 
\end{proof}

\bibliographystyle{amsalpha}
\bibliography{approx}

\newcommand{\etalchar}[1]{$^{#1}$}
\providecommand{\bysame}{\leavevmode\hbox to3em{\hrulefill}\thinspace}
\providecommand{\MR}{\relax\ifhmode\unskip\space\fi MR }
\providecommand{\MRhref}[2]{%
  \href{http://www.ams.org/mathscinet-getitem?mr=#1}{#2}
}
\providecommand{\href}[2]{#2}
\begin{thebibliography}{DLTW08}

\bibitem[ALM{\etalchar{+}}98]{AroLunMotSudSze98JACM}
Sanjeev Arora, Carsten Lund, Rajeev Motwani, Madhu Sudan, and Mario Szegedy,
  \emph{Proof verification and the hardness of approximation problems}, JACM
  \textbf{45} (1998), no.~3, 501--555.

\bibitem[AS98]{AroSaf98JACM}
Sanjeev Arora and Shmuel Safra, \emph{Probabilistic checking of proofs: A new
  characterization of {NP}}, JACM \textbf{45} (1998), no.~1, 70--122.

\bibitem[CLS16]{CLS16}
Richard Cleve, Li~Liu, and William Slofstra, \emph{Perfect commuting-operator
  strategies for linear system games}, Journal of Mathematical Physics (2016),
  to appear (arXiv:1606.02278).

\bibitem[CM14]{CM14}
Richard Cleve and Rajat Mittal, \emph{Characterization of {Binary} {Constraint}
  {System} {Games}}, Automata, {Languages}, and {Programming}, Lecture {Notes}
  in {Computer} {Science}, no. 8572, Springer Berlin Heidelberg, 2014,
  arXiv:1209.2729, pp.~320--331.

\bibitem[CV15]{CV15}
Matthew Coudron and Thomas Vidick, \emph{Interactive proofs with approximately
  commuting provers}, International Colloquium on Automata, Languages, and
  Programming (ICALP 2015), 2015, pp.~355--366.

\bibitem[DLTW08]{DLTW08}
Andrew~C. Doherty, Yeong-Cherng Liang, Benjamin Toner, and Stephanie Wehner,
  \emph{The quantum moment problem and bounds on entangled multi-prover games},
  Proc. 23rd IEEE Conf. on Computational Complexity (CCC'08), IEEE Computer
  Society, 2008, pp.~199--210.

\bibitem[FJVY18]{FJVY18}
J.~Fitzsimons, Z.~Ji, T.~Vidick, and H.~Yuen, \emph{Quantum proof systems for
  iterated exponential time, and beyond}, preprint (2018), arXiv:1805.12166.

\bibitem[Ger93]{Ge93}
SM~Gersten, \emph{Isoperimetric and isodiametric functions of finite
  presentations}, Geometric group theory: Proceedings of the Symposium held in
  Sussex 1991, Volume 1 (Graham~A. Niblo and Martin~A. Roller, eds.), London
  Mathematical Society Lecture Note Series 181, vol. 181, Cambridge University
  Press, 1993, pp.~79--97.

\bibitem[IKP{\etalchar{+}}08]{IKPSY08}
T.~Ito, H.~Kobayashi, D.~Preda, X.~Sun, and A~C.-C. Yao, \emph{Generalized
  {T}sirelson inequalities, commuting-operator provers, and multi-prover
  interactive proof systems}, 23rd Annual IEEE Conference on Computational
  Complexity (CCC 2008), 2008.

\bibitem[IV12]{IV12}
Tsuyoshi Ito and Thomas Vidick, \emph{A multi-prover interactive proof for
  {NEXP} sound against entangled provers}, Foundations of Computer Science
  (FOCS), 2012 IEEE 53rd Annual Symposium on, IEEE, 2012, pp.~243--252.

\bibitem[Ji17]{Ji16}
Zhengfeng Ji, \emph{Compression of quantum multi-prover interactive proofs},
  Proceedings of the 49th Annual ACM SIGACT Symposium on Theory of Computing
  (New York, NY, USA), STOC 2017, ACM, 2017, pp.~289--302.

\bibitem[LS77]{LS77}
R.C. Lyndon and P.E. Schupp, \emph{Combinatorial group theory}, Classics in
  Mathematics, Springer, 1977.

\bibitem[NPA07]{NPA07}
Miguel Navascu{\'{e}}s, Stefano Pironio, and Antonio Ac{\'{\i}}n,
  \emph{Bounding the set of quantum correlations}, Phys. Rev. Lett. \textbf{98}
  (2007), 010401.

\bibitem[NPA08]{NPA08NJP}
Miguel Navascu{\'{e}}s, Stefano Pironio, and Antonio Ac{\'{\i}}n, \emph{A
  convergent hierarchy of semidefinite programs characterizing the set of
  quantum correlations}, New Journal of Physics \textbf{10} (2008), no.~073013.

\bibitem[Oza13]{Oz13b}
Narutaka Ozawa, \emph{Tsirelson's problem and asymptotically commuting unitary
  matrices}, Journal of Mathematical Physics \textbf{54} (2013), no.~3, 032202.

\bibitem[SBR02]{SBR02}
Mark~V. Sapir, Jean-Camille Birget, and Eliyahu Rips, \emph{Isoperimetric and
  {Isodiametric} {Functions} of {Groups}}, Annals of Mathematics \textbf{156}
  (2002), no.~2, 345--466.

\bibitem[Slo16]{Sl16}
William Slofstra, \emph{Tsirelson's problem and an embedding theorem for groups
  arising from non-local games}, preprint (arXiv:1606.03140).

\bibitem[SV18]{SV17}
William Slofstra and Thomas Vidick, \emph{Entanglement in non-local games and
  the hyperlinear profile of groups}, Annales {H}enri {P}oincare \textbf{19}
  (2018), no.~10, 2979--3005.

\end{thebibliography}

\end{document}